\documentclass[11pt]{article}

\usepackage{times}
\usepackage{graphicx}
\usepackage{psfrag}
\usepackage{amsmath,amsthm}
\usepackage{amsfonts}

\theoremstyle{plain} 
\newtheorem{thm}{Theorem}[section]

\theoremstyle{definition}
\newtheorem{defn}{Definition}[section]

\theoremstyle{remark}

\begin{document}

\parskip=10pt

\flushbottom 

\title{The complexity of bit retrieval} 

\author{
Veit Elser\\
Department of Physics\\
Cornell University} 

\date{}

\maketitle

\begin{abstract}
Bit retrieval is the problem of reconstructing a binary sequence from its periodic autocorrelation, with applications in cryptography and x-ray crystallography. After defining the problem, with and without noise, we describe and compare various algorithms for solving it. A geometrical constraint satis\-faction algorithm, relaxed-reflect-reflect, is currently the best algorithm for noisy bit retrieval.
\end{abstract}

\section{Bit retrieval}

The \textit{bit retrieval} problem is like the problem of factoring integers, but with some modifications to the rules of arithmetic. These modifications are illustrated below, where the calculation of $13\times 19$ in base-2 is contrasted with the rules of bit retrieval. There are two changes: (i) exponents in the binary expansion are periodic (columns ``wrap around"), and (ii) the columns are summed without carrying. In both factoring and bit retrieval the problem is to reverse the process: find the binary sequences at the top, given their product at the bottom.

\begin{table}[h!]
\centering
\begin{tabular}{lllllllll}
& & & & &1&1&0&1\\
$\times$& & & &1&0&0&1&1\\
 \hline
& & & & &1&1&0&1\\
& & & &1&1&0&1& \\
&1&1&0&1& & & & \\
\hline
&1&1&1&1&0&1&1&1\\
 \end{tabular}
 \qquad\qquad
 \begin{tabular}{llllll}
&0&1&1&0&1\\
$\times$&1&0&0&1&1\\
 \hline
&0&1&1&0&1\\
&1&1&0&1&0\\
&1&0&1&1&0\\
\hline
&2&2&2&2&1\\
 \end{tabular}
 
\begin{quote}
The rules of ordinary base-2 multiplication (left) are modified (right) so that exponents have period 5 and columns are summed without carries.
\end{quote}

\end{table}

Periodic or ring-like arrangements of integers that are combined with the rules just described are elements of the polynomial ring $Z_N=\mathbb{Z}[x]/(x^N-1)$, where $N$ is the period of the exponents ($N=5$ in the example). 
Factoring elements in $Z_N$ is hard, even when we are told the coefficients of the factors are limited to 0 and 1. We will see shortly that in bit retrieval it makes more sense to instead limit the coefficients to $\pm 1$; we therefore define the set
\begin{equation}
S_N=\{s_0+s_1 x+\cdots +s_{N-1}x^{N-1}\in Z_N\colon s_k=\pm 1, \quad 0\le k\le N-1\}.
\end{equation}

``Retrieval" is a reference to \textit{phase retrieval}, an important special case where the two coefficient sequences  are reflections of each other (one ``ring" is the mirror of the other). This brings us to our first formulation of bit retrieval:
\begin{defn}
\textit{Bit retrieval} is the problem where, given $a(x)\in Z_N$ known to have the form $a(x)=s(x)s(1/x)$ for some $s(x)\in S_N$, we must find $s'(x)\in S_N$ such that $s'(x)s'(1/x)=a(x)$.
\end{defn}
The problem definition sidesteps the question of uniqueness. Clearly, if $s'(x)$ is a solution, then so are $\pm s'(x)x^r$ and $\pm s'(1/x)x^r$, for arbitrary $r$. We restricted the solutions of bit retrieval to be elements of $S_N$ so that they form orbits in this group of order $4N$. Another nice property of $\pm 1$ sequences that will be useful later (section \ref{sec:convex}) is that they have the same 2-norm.

The fastest known algorithm for bit retrieval, as defined above, was discovered by How\-grave-Graham and Szydlo \cite{HGS} and was based on earlier work by Gentry and Szydlo \cite{GS} that proposed an attack on the NTRU digital signature scheme. This algorithm has about the same complexity as factoring a number of $O(N\log{N})$ bits; in fact, the first and hardest step of the algorithm is precisely the factorization of a number of that size. However, a seemingly small change can make bit retrieval much harder than this. Before we describe this change, we further develop the relationship to phase retrieval.

Bit retrieval is a highly idealized model of phase retrieval in crystallography. In that setting, the polynomial coefficients $s_0,\ldots,s_{N-1}$ are samples within one period of a periodic function (a 1D crystal), and the coefficients of the product $a(x)=s(x)s(1/x)$ their (periodic) autocorrelation:
\begin{equation}\label{auto}
a_k=\sum_{l=0}^{N-1} s_l s_{l-k}=\sum_{l=0}^{N-1} s_l s_{l+k}=a_{-k}.
\end{equation}
The indices in \eqref{auto} are all taken mod $N$. In crystallography one would refer to $s$ as the \textit{contrast} because one acquires information about it through its action on radi\-ation to produce diffraction patterns. Since the contrast elements are all $\pm 1$, the central autocorrelation is trivial: $a_0=N$. There are $\lfloor N/2\rfloor$ nontrivial autocorrelations as a result of the reflection symmetry in \eqref{auto}.

To complete the connection to phase retrieval, we start with the identity
\begin{equation}\label{fouriermag}
\hat{a}_q=\frac{1}{\sqrt{N}}\sum_{k=0}^{N-1}e^{i2\pi k q/N}a_k=\sqrt{N}\,|\hat{s}_q|^2,
\end{equation}
where
\begin{equation}\label{FT}
\hat{x}_q=\frac{1}{\sqrt{N}}\sum_{k=0}^{N-1}e^{i2\pi k q/N}x_k
\end{equation}
defines the Fourier transform of a periodic sequence $x$. From \eqref{fouriermag} we see that the autocorrelations only give us the Fourier transform magnitudes $|\hat{s}|^2$, and without knowledge of the phases of $\hat{s}$ we are unable to invert the transform \eqref{FT} to recover the signs $s$. Phase retrieval refers to the strategy of discovering the unknown phases of $\hat{s}$ by demanding consistency with additional information we have about the contrast $s$. In the case of bit retrieval, this translates to the observation that only very special sets of phases, when combined with the known magnitudes, produce contrast values comprising only $\pm 1$.

To a mathematician, the Fourier magnitudes $|\hat{s}|^2$ are algebraic numbers that when examined in great enough detail will reveal the $\pm 1$ coefficients of $s$, even without knowledge of the phases. By contrast, when these same numbers are measured in a diffraction experiment, they are known only to a finite precision. Given that we are retrieving elements from a finite set, how much imprecision or noise can be tolerated?

The observation that the autocorrelation coefficients are always integers in the same congruence class mod 4 as $N$ (appendix \ref{sec:hadamard}) motivates the following set of symmetric polynomials as the smallest ``quanta" of autocorrelation noise:
\begin{equation}\label{noise}
E_N=\left\{e(x)\in Z_N\colon e_0=0;\; e_k=e_{-k}=\pm 2,\; 1\le k\le \lfloor N/2\rfloor\right\}.
\end{equation}
\begin{defn}
\textit{Noisy bit retrieval} is the problem where, given a noisy autocorrelation $n(x)$ known to have the form $n(x)=s(x)s(1/x)+e(x)$, where $s(x)\in S_N$ and $e(x)\in E_N$, we must find elements $s'(x)\in S_N$ and $e'(x)\in E_N$  such that $n(x)=s'(x)s'(1/x)+e'(x)$.
\end{defn}
In noisy bit retrieval we know that all of the $k\ne 0$ noisy autocorrelations $n_k$ are off by $\pm 2$, but we do not know whether the true autocorrelations are obtained by rounding up by 2 or down by 2. We will see in the next section how even this amount of noise in the data completely undermines the most efficient bit retrieval algorithms. The fastest known algorithms for noisy bit retrieval have complexity $2^{c N}$; algorithms and estimates of the constant $c$ are discussed in sections \ref{sec:convex} and \ref{sec:RRR}.

We would still like it to be true that the introduction of noise \eqref{noise} does not, with high probability, sacrifice solution uniqueness. This is supported by the following theorem. Here we assume uniform probability distributions, both on the set of sign sequences $S_N$ and the noise $E_N$ we apply to their autocorrelations; ``random" elements of these sets are elements sampled from the uniform distribution.

\begin{thm}
Let $s$ and $e$ be random elements respectively of $S_N$ and $E_N$, and $n(x)=s(x)s(1/x)+e(x)$ the corresponding noisy autocorrelation. Let $s'$ be the same as $s$ but with a single one of its $N$ signs reversed. The probability, that there exists an $e'\in E_N$ such that $s'(x)s'(1/x)+e'(x)=n(x)$ (the modified $s'$ is also compatible with $n$), is equal to $(3/4)^{(N-1)/2}$ for odd $N$ and $(1/2)(3/4)^{N/2-1}$ for even $N$.
\end{thm}
\begin{proof}
First consider odd $N$. Without loss of generality we may assume the reversed sign is in position 0, so $s'_0=-s_0$, and $s'_k=s_k$ for $k\ne 0$. Upon reversal, the autocorrelations change as follows:
\begin{equation}
a'_k-a_k=(s'_0-s_0)(s_k+s_{-k})=\pm 2(s_k+s_{-k}),\quad 1\le k\le(N-1)/2.
\end{equation}
Each change arises from a pair of independent signs, and is therefore 0 with probability $1/2$ and $\pm 4$ with probability $1/2$. Since the noisy data $n$ has not changed, $e'_k-e_k= a_k-a'_k$. Whenever $a'_k-a_k=0$, an unchanged $e'_k=e_k$ is compatible with the modified $s'$. However, whenever $a'_k-a_k=\pm 4$, both $e_k$ and $e'_k$ are determined (their values are limited to $\pm 2$) and in particular, only one choice of $e_k$ allows for $s'$ to be compatible with $n$. The net probability that there exists a compatible $e'_k$ is therefore $(1/2)(1)+(1/2)(1/2)=3/4$ and the stated result follows from the independence of the $(N-1)/2$ outcomes for the different $k$. When $N$ is even, only the case $k=N/2$ is changed because the change in $a_{N/2}$ is $\pm 4$ with probability 1. The probability that there exists a compatible $e'_{N/2}$ therefore changes from $3/4$ to $1/2$.
\end{proof}

We conjecture that uniqueness in the sense of the above theorem extends beyond the simple case of a single reversed sign. There is always non-uniqueness stemming from the invariance of the autocorrelation with respect to the order $4N$ group generated by cyclic-shifts, reflection and sign reversal. But this is a small group and inconsequential if we view bit retrieval, in information-theoretic terms, as the decoding stage of a noisy communication channel. In the ``noisy autocorrelator channel" an input signal of $N$ bits, in the form of signs $s$, is encoded with noise as $s\to a+e=n$. The stronger result suggested by the theorem is that the information capacity of this channel (an asymptotic property for large $N$) is the same as the entropy of the uniform distribution on the inputs.

Turning the (probabilistically qualified and symmetry amended) uniqueness conjecture into a theorem presents difficult challenges.  There exist polynomials, for example ($N=13$)
\begin{eqnarray}
s(x)&=&1+x^2+x^3+x^4+x^5+x^6+x^7+x^{10}+x^{11}\\
s'(x)&=&1+x^2+x^3+x^4+x^7+x^9+x^{10}+x^{11}+x^{12},
\end{eqnarray}
that have the same autocorrelation and yet are not in the same orbit of the order $4N$ symmetry group. The equality of the autocorrelations is in this case explained by the fact that
\begin{equation}
s(x)=p(x)q(x)=(1+x^2+x^7)(1+x^3+x^4),
\end{equation}
and $s'(x)=p(x)q(1/x)$. To prove the theorem one needs to bound this form of non-uniqueness, which exists even without noise. Though extremely rare, the phenomenon of factorizable solutions (contrast) is also known to occur in crystal\-lography\cite{PS}. 

By the standards of crystallography, the noise defined by \eqref{noise} has an unrealistic dependence on $N$. Since the noise coefficients $e$ are $O(1)$, so will be the difference in the Fourier coefficients $\hat{e}$, between the true and noisy transforms, $\hat{a}$ and $\hat{n}$. By \eqref{fouriermag} this translates into $O(1/\sqrt{N})$ errors in the Fourier magnitudes $|\hat{s}|^2$. Crystallography experiments are noisier, being content with an $O(1)$  signal-to-noise ratio and therefore noise amplitudes for $e$ and $\hat{e}$ in the bit retrieval model growing as $O(\sqrt{N})$. This brings us to yet a third problem: 
\begin{defn}
\textit{Fixed-precision bit retrieval} is the problem where, given precision $\eta>0$ and a noisy autocorrelation transform $\hat{n}$ known to satisfy
\begin{equation}\label{finiteprecision}
\left||\hat{s}_q|^2-\hat{n}_q/\sqrt{N}\right|<\eta,\quad 0\le q\le \lfloor N/2\rfloor,
\end{equation}
for some Fourier transformed sequence of signs $s$, we must find such a sign sequence.
\end{defn}
This version of bit retrieval comes closest to the phase retrieval problem in crystallography. There the data naturally arrives via the Fourier transform and is always subject to noise. The order $N$ of the cyclic group in bit retrieval corresponds to the number of resolution elements, or voxels, in the representation of the contrast. That the symmetry group of the 3D problem is not the cyclic group of order $N$, but a direct product of three such groups having the same order, is probably largely irrelevant to the complexity of bit retrieval. Finally, although a strict two-valued contrast is a poor way to approximate a continuous contrast function (electron density), it is not a bad model for representing a dilute collection of equally scattering atoms at low resolution.

The fastest algorithms for solving fixed-precision bit retrieval, like noisy bit retrieval, have complexity $2^{c N}$, where the constant $c$ now depends on the noise parameter $\eta$. But unlike the noisy version, solutions in the fixed-precision version have extensive entropy, that is, grow in number exponentially with $N$. This can be argued non-rigorously as follows.

Take $N$ large and consider flipping a large random subset of $M$ signs, while keeping $M\ll N$. The Fourier transform changes as $\hat{s}_q\to \hat{s}_q + \Delta\hat{s}_q$ where, using a result of Freedman and Lane \cite{FL}, the $\Delta\hat{s}_q$ are independent complex-normal random variables with zero mean  and variance $O(M/N)$. The probability that the flips violate any of the corresponding inequalities in \eqref{finiteprecision}, in the limit of small variance, is an integral over the tail of a Gaussian distribution and depends on $\eta$ as $B_q\exp{(-b_q N\eta^2/M)}$ for some positive constants $B_q$ and $b_q$. The probability that no inequality is violated behaves as
\begin{equation}\label{noviolation}
\prod_{q=0}^{\lfloor N/2\rfloor} \left(1-B_q\exp{(-b_q N\eta^2/M)}\right).
\end{equation}
Now consider the limit $N\to\infty$ with $M/N$ held fixed and $M/N\ll \eta^2$. In this limit \eqref{noviolation} approaches 1 for any of the sets of flipped signs which, for fixed $M/N$, have extensive entropy. Solutions therefore have extensive entropy for any $\eta>0$. Crystallography with fixed $\eta$ can escape this source of non-uniqueness by keeping $N$ under a bound proportional to $1/\eta^2$.

Fixed-precision bit retrieval would reduce to noise-free bit retrieval if instead of fixing $\eta$ (as $N$ increases) we were allowed to take the limit $\eta\to 0$. In this limit, the Fourier transform of the noisy $\hat{n}$, after rounding the coefficients, is the autocorrelation $a$ of bit retrieval. We get a variant of noisy bit retrieval if instead we take limits such that $N\eta^2=O(1)$, i.e. keeping $\eta$ just small enough to preserve solution uniqueness. The fixed-precision version lends itself naturally to  geometrical constraint satisfaction algorithms, two of which we shall describe in detail. Unlike the algebraic algorithms developed for solving the noise-free problem, the geometric algorithms are easily adapted to solve any of the three problems. Not surprisingly, the complexity of the geometric algorithms is relatively insensitive to $\eta$.

\section{Symmetry and noise}\label{sec:noise}

To appreciate the effect of noise on bit retrieval complexity, we focus in this section on instances where it is known that the signs $s$ have a reflection symmetry. We are then free to target the rotated polynomial $s'(x)=x^r s(x)$ that has the property $s'(x)=s'(1/x)$. To avoid complications in the presentation that do not alter the main ideas, we restrict ourselves to prime $N$ in this section.

Symmetric bit retrieval is very easy. Dropping the prime on our reflection symmetric signs, we define $b(x)\in Z_N$ with coefficients $b_k=(1-s_k)/2\in \{0,1\}$. The coefficients $a'_k$ of the corresponding autocorrelation $a'(x)=b(x)b(1/x)=b(x)^2$ are related to the sign autocorrelations as follows:
\begin{equation}
a'_k=\frac{1}{4}\left(a_k+N-2\sum_l s_l\right).
\end{equation}
By \eqref{auto} the sum of the signs is one of the square roots of the sum of the $a_k$'s. Exercising symmetry to always select the non-negative root, the transformed $a'(x)\in Z_N$ is known. We now observe there are exactly as many bits of information in the symmetric $b(x)$ as there are parity bits in the $a'(x)$ coefficients. This suggests reducing all the coefficients mod 2, so we are working in the ring $(\mathbb{Z}/2)[x]/(x^N-1)$:
\begin{eqnarray}
a'(x)&=&\left(b_0 + \sum_{k=1}^{(N-1)/2}b_k(x^k+x^{-k})\right)^2\\
&=&b_0 + \sum_{k=1}^{(N-1)/2}b_k(x^{2k}+x^{-2k}).
\end{eqnarray}
Since $N$ is a prime greater than 2, there is a unique element $2^{-1}$ in the field of $N$ elements and an explicit formula for bit retrieval:
\begin{equation}
b_k=a'_{2^{-1}k}\pmod{2}.
\end{equation}

By reducing the transformed autocorrelation coefficients mod 2 we have made our bit retrieval algorithm maximally vulnerable to noise. Indeed, with the $\pm 2$ uncertainty in $a_k$ of noisy bit retrieval, the parities of the $a'_k$'s are completely uncertain and so it would seem, the bits $b_k$. However, we next consider a more elaborate polynomial-time algorithm whose noise tolerance is somewhat better. Since sign reversal $s\to -s$ is the only symmetry remaining in reflection symmetric bit retrieval, it is not surprising that this special case of the problem can be reduced to a shortest lattice vector problem, which shares this symmetry.

When the signs $s$ have reflection symmetry, from \eqref{FT} we see that $\hat{s}$ is purely real and we can write $\hat{s}_q= |\hat{s}_q| y_q$, where $y_q=\pm 1$. We then have the following equations relating the signs $x_k$ of the unknown sequence and the unknown signs $y_q$ of its Fourier transform:
\begin{equation}\label{symFT}
\sqrt{N}\,|\hat{s}_q| \,y_q=x_0+\sum_{k=1}^{(N-1)/2}2\cos{(2\pi k q/N)}\,x_k,\quad 0\le q\le (N-1)/2.
\end{equation}
Since $y_{-q}=y_q$, there are just as many independent equations and Fourier signs, $M=(N+1)/2$, as there are unknown signs in the reflection-symmetric sequence we are attempting to retrieve.

The observation that \eqref{symFT} should hold for arbitrary levels of precision, in numerical approximations of the cosine functions and the data $\sqrt{N}\,|\hat{s}_q|$, leads to a polynomial time bit retrieval algorithm for sequences known to have reflection symmetry. Unlike the $q=0$ equation, which only reveals the number of $+1$ signs (up to overall sign reversal), the other equations, individually, become nontrivial instances of the integer partitioning problem when their coefficients are multiplied by a large number $K=2^P$ and then rounded to the nearest integer. Unlike the usual integer partitioning problem, here we require only that the partition produces a sum consistent with the round-off errors. Nevertheless, it is easy to produce arbitrarily good approximate integer equations because the round-off has a fixed bound while arbitrarily large $P$-bit approximations of the Fourier coefficients can be computed in time that grows as a polynomial in $P$.

It is straightforward to adapt the method of Lagarias and Odlyzko \cite{LO}, for solving low density subset sum problems, to solve symmetric bit retrieval. Low density in our context corresponds to setting the number of bits $P$, in the approximation of the coefficients, sufficiently large in comparison to the number of unknown signs, $2M$. However, rather than use just one of the $q\ne 0$ equations, and the information in just one of the Fourier magnitudes, we construct a lattice $\Lambda$ from information provided by all $M$ equations. The generators of $\Lambda$ are the rows of the 
following $2M\times 2M$ matrix:
\begin{equation}\label{G}
G=\left[
\begin{array}{cc}
\lfloor K D\rceil & 0\\
\lfloor K C\rceil & I_{M\times M}
\end{array}
\right],
\end{equation}
where $\lfloor\, \cdots \rceil$ denotes rounding to the nearest integer and the $M\times M$ blocks $C$ and $D$ are defined by
\begin{equation}
C_{k q}=\left\{
\begin{array}{ll}
1, & k=0\\
2\cos{(2\pi k q/N)}, & 1\le k\le M-1,
\end{array}
\right.
\end{equation}
\begin{equation}
D=\mathrm{diag}\left(-\sqrt{N}|\hat{s}_0 |\;,\ldots,\;-\sqrt{N}|\hat{s}_{M-1} |\right).
\end{equation}
By construction, $\Lambda$ has two short vectors:
\begin{eqnarray}\label{short}
v_1&=&[w_0\cdots w_{M-1}\;\; 1\cdots 1]=[0\cdots 0\;\; 1\cdots 1]\cdot G\\
v_s&=&[z_0\cdots z_{M-1}\;\; s_0\cdots s_{M-1}]=[y_0\cdots y_{M-1}\;\; s_0\cdots s_{M-1}]\cdot G,\label{vs}
\end{eqnarray}
where $w_0,\ldots,w_{M-1}$ and $z_0,\ldots,z_{M-1}$ are sets of small integers produced by round-off. Vector $v_1$ is small because each of the columns of the matrix $C$ has zero sum while $v_s$ is small by equations \eqref{symFT}.

Each column of \eqref{vs} represents one instance of the integer partitioning problem: assigning $M+1$ signs to the same number of $P$-bit integers to produce a small sum. In the equivalent subset sum problem we must find a subset of $M+1$ $P$-bit integers to produce a given target sum, again with neglect of the low order round-off bits. In base-2 arithmetic, a solution is checked by verifying that nearly $P$ column sums (low order bits excepted) are all even, where these were equally likely to have been either parity in a randomly guessed subset. Reasoning probabilistically, we conclude that $P$ must be at least as large as $M$ if we expect to recover a unique subset, or choice of signs in the equivalent integer partitioning problem. Fewer bits should suffice when the same set of signs is required to solve all $M-1$ non-trivial ($q\ne 0$) integer partitioning problems represented by \eqref{vs}. In fact, the necessary number of bits would be bounded if the information provided by each partitioning problem is in some sense independent of the others.

The question of how to efficiently find a partition places different demands on the number of bits in our integer approximation of the symmetric bit retrieval problem. The Lagarias-Odlyzko algorithm, associated with a single one of our $M-1$ partitioning problems, and assuming the specific non-random integers in $G$ are well modeled by average-case behavior, requires $P=O(M^2)$. We have not attempted to extend the analysis of the algorithm to the generator matrix $G$, and instead have performed experiments with the symmetric Hadamard sequence instances, defined in  appendix \ref{sec:hadamard}, that we believe to be among the hardest. In each experiment we apply the \textit{Mathematica} implementation of the LLL lattice reduction algorithm \cite{LLL} to $G$ and record a success when among the reduced basis we find a vector $a v_1+b v_s$, where $b\ne 0$.

In Figure 1 we show the dependence of the bit length $P$ on successful retrieval of symmetric Hadamard sequences by LLL basis reduction of the generator matrix $G$ given in \eqref{G}. At each $N$ for which such a sequence exists we plot the smallest $P$ for which the retrieval was successful. We see that $P$ appears to grow linearly with $N$. This behavior is below the quadratic growth required by LLL for solving a single random subset sum problem of the same size, but well above the probabilistically argued bounded bit-length required for solution uniqueness. 

\begin{figure}[!t]
\begin{center}
\includegraphics[width=4.in]{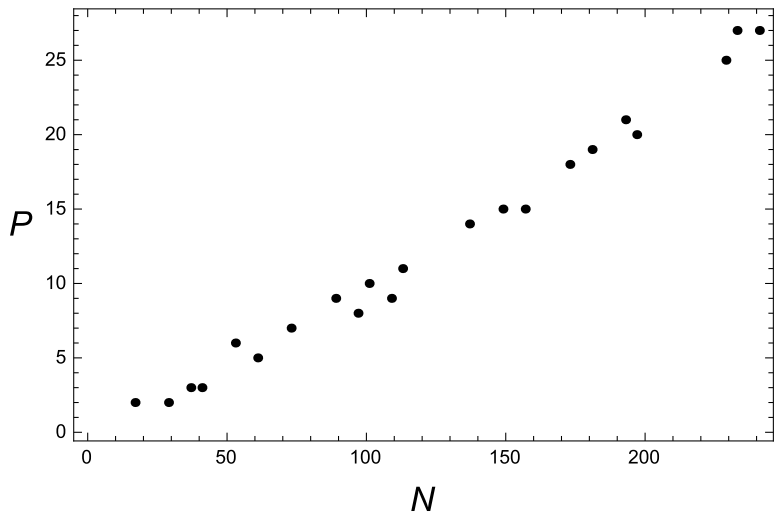}
\end{center}
\caption{Growth in the minimum bit length $P$ required for successful bit retrieval of symmetric Hadamard sequences of length $N$ by the LLL basis reduction algorithm.}
\end{figure}

Whereas any polynomial growth of the required bit length of the Fourier magnitude ``data" is consistent with a polynomial-time algorithm, that the data precision must grow at all 
eliminates the fixed-precision variant of bit retrieval. We will use the symmetric Hadamard instances to show that the lattice basis reduction algorithm also fails to solve noisy bit retrieval.

Symmetric Hadamard sequences (appendix \ref{sec:hadamard}) exist for all prime lengths $N\equiv 1\pmod{4}$ and are interesting because of their low autocorrelations, $a_k\in\{-3,1\}$, \mbox{$k\ne 0$}. For any such sequence there exists an $e\in E_N$ such that the noisy autocorrelation $n=a+e$ has $n_k=-1$, $k\ne 0$. The smallness of the $k\ne 0$ autocorrelations, relative to $a_0=N$, has by \eqref{fouriermag} the effect that the $q\ne 0$ Fourier magnitudes are nearly equal. Using the fact that the symmetric Hadamard sequences have the property $\hat{s}_q=s_q+1/\sqrt{N}$, their Fourier magnitudes take two values:
\begin{equation}\label{trueMag}
\hat{a}_q/\sqrt{N}=|\hat{s}_q|^2=\left|\pm 1+\frac{1}{\sqrt{N}}\right|^2,\qquad q\ne 0.
\end{equation}
When this noise-free data for a symmetric instance of bit retrieval is given as input to the LLL basis reduction algorithm above, we can retrieve the sequence in polynomial time. But suppose we are instead given the noisy data, where $n_k=-1$, $k\ne 0$. What should we expect of the LLL method when the Fourier magnitudes,
\begin{equation}\label{approxMag}
\hat{n}_q/\sqrt{N}=1-\frac{1}{N},\qquad q\ne 0
\end{equation}
obtained from $n$ are used instead of \eqref{trueMag}? Experiment shows that the algorithm in this case succeeds only up to $N=73$, and then only up to a maximum bit length $P_\mathrm{max}$. For $P>P_\mathrm{max}$ the difference between \eqref{trueMag} and \eqref{approxMag} has grown to be so large in the basis $G$ that $v_s$, defined by \eqref{vs}, is no longer a short vector (and found by LLL); in particular, $P_\mathrm{max}=5$ when $N=73$. With the addition of noise our polynomial-time algorithm did not simply suffer a decline in performance: it ceased being an algorithm altogether.

We expect a similar vulnerability to noise of the best known algorithm for non-symmetric instances of bit retrieval. This is the algorithm developed by Howgrave-Graham and Szydlo \cite{HGS}, recently generalized by Lenstra and Silverberg \cite{LS} to other groups. The first step of the HGS algorithm calls for the prime factorization of the product
\begin{equation}\label{norm}
\prod_{q=1}^{N-1}\left(\sqrt{N}\; \hat{a}_q\right)=\prod_{q=1}^{N-1}\left(N \; |\hat{s}_q|^2\right),
\end{equation}
a perfect square integer whose square root coincides with the norm of $s$ when the latter is viewed as an integer in the cyclotomic number field. This first step dominates the complexity, growing as $\exp{O(M^{1/3})}$ when using the number field sieve to factor $M$-bit norms, where $M=O(N\log{N})$. Still, this growth is far superior to any known algorithm that tolerates noise. However, were we to replace $\hat{a}$ in \eqref{norm} with its noisy counterpart $\hat{n}$, we have a problem because the resulting integer already fails at being a perfect square\footnote{While there is noise in cryptography too, this was not the case for the intended application of HGS, where by averaging sufficiently many digital signatures and rounding, one has access to the true norm.}.

\section{Hardness}\label{sec:hardness}

The product of the Fourier magnitudes of the sign sequence being retrieved has already made two appearances. First, in the lattice basis reduction algorithm for symmetric instances, we see from the generator matrix \eqref{G} that this product is the sequence-dependent factor in the lattice determinant. The second appearance was in the cyclotomic-integer norm \eqref{norm} that has to be factored in the sub-exponential-time HGS algorithm as a first step in solving general instances. In the presence of noise, when neither of these algorithms can be used, we will see that the exponential-time algorithms that take their place have a similarly strong dependence on the Fourier magnitude product. These observations motivate the following definition of a hardness index:
\begin{defn}\label{hardness}
The bit retrieval hardness index $h(s)$ of a sequence $s$ is the geometric mean of the Fourier magnitudes
\begin{equation}
\left\{|\hat{s}_q|^2\colon q\not\equiv -q\pmod{N}\right\},
\end{equation}
that is, those where the corresponding phases are not restricted to $\pm 1$. 
\end{defn}
The index is useful because it is easily computed from the bit retrieval input. It vanishes if any one of the magnitudes vanishes. This is appropriate, since the magnitude constraints for those $q$ can then be replaced by simple equality constraints. When $N$ is odd, the arithmetic-geometric mean inequality gives the following upper bound:
\begin{eqnarray}
h(s)&\le&\frac{1}{N-1}\sum_{q=1}^{N-1} |\hat{s}_q|^2\\
&=&\frac{N-|\hat{s}_0|^2}{N-1}\\
&\le&\frac{N-(1/N)}{N-1}=1+\frac{1}{N}.
\end{eqnarray}
The two means are equal if and only if all the $q\ne 0$ Fourier magnitudes are equal; this and $|\hat{s}_0|=1/\sqrt{N}$ may be taken as the defining properties of the perfect Hadamard sequences (appendix \ref{sec:hadamard}). The latter require $N\equiv 3 \pmod{4}$; when $N\equiv 1 $ the slightly imperfect, though symmetric Hadamard sequences that take their place have two-valued $q\ne 0$ Fourier magnitudes and the slightly smaller index $h(s)=1-1/N$. When $N$ is even the estimate must be modified since $|\hat{s}_0|$ can vanish; the resulting upper bound is then $1+2/(N-2)$.

The relationship of the hardness index to the complexity of constraint satisfaction algorithms is particularly direct. Consider the autocorrelation constraint set $A(a)\subset \mathbb{R}^N$ defined (for general $N$) by
\begin{equation}
A(a)=\left\{ x\in \mathbb{R}^N\colon |\hat{x}_q|^2=\hat{a}_q/\sqrt{N},\quad 0\le q\le \lfloor N/2\rfloor\right\}.
\end{equation}
Geometrically $A(a)$ is the Cartesian product of pairs of points, associated with $q=0$ and also $q=N/2$ when $N$ is even, as well as $\lfloor(N-1)/2\rfloor$ circles associated with the remaining $q$ for which $q\not\equiv -q\pmod{N}$.
A reasonable general strategy for bit retrieval is to devise a scheme that systematically samples $A(a)$ at some resolution, identifying promising candidates by the proximity of their coordinates to $\pm 1$. With uniform sampling, the number of sample points at fixed resolution will be proportional to the volume of $A(a)$:
\begin{equation}\label{vol}
\mathrm{vol}(A(a))\propto h(s)^{\frac{\lfloor(N-1)/2\rfloor}{4}}.
\end{equation}

The RRR constraint satisfaction algorithm (section \ref{sec:RRR}) generates samples by iterating a map $\mathbb{R}^N\to \mathbb{R}^N$ constructed from the projection to $A(a)$ and also the projection to the hypercube,
\begin{equation}
B=\left\{ x\in \mathbb{R}^N\colon x_k=\pm 1,\quad 0\le k\le N-1\right\}.
\end{equation}
Figure 2 shows the behavior of the RRR iteration count when solving random $N=101$ instances selected for five values of the hardness index. The $h\approx 1$ symmetric Hadamard sequence is too difficult for these experiments, but if we include an extrapolation from shorter symmetric Hadamard sequences that can be solved (Fig. 15), the iteration count has an upper range near $2\times 10^{11}$. The RRR algorithm is currently the best known for bit retrieval with noise. Since the volume \eqref{vol} for $N=101$ changes by $1.6\times 10^{9}$ as $h$ ranges between $0.3$ and $0.7$, the RRR algorithm must be doing something better than uniformly sampling the set $A(a)$ because the iteration count between those extremes changes by only $3.5\times 10^{4}$.

\begin{figure}[!t]
\begin{center}
\includegraphics[width=4.in]{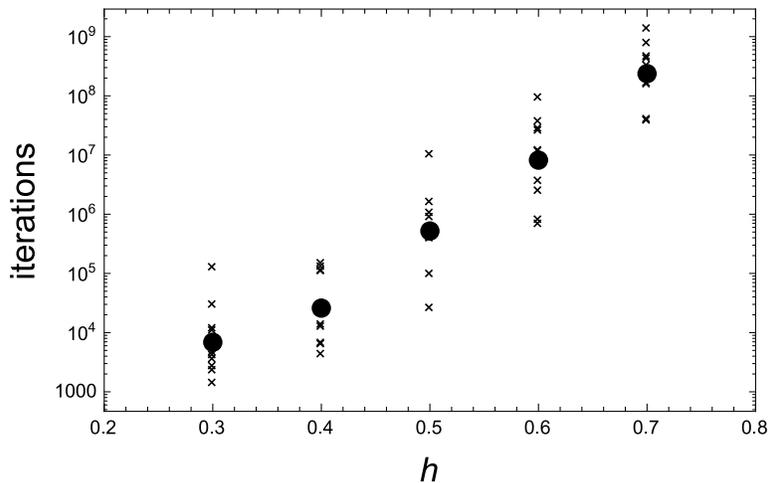}
\end{center}
\caption{Growth in the iteration count of the RRR constraint satisfaction algorithm (section \ref{sec:RRR}) with the hardness index $h$. Shown are data for $N=101$, 10 random instances at each of five values of the index. The circular markers give the geometric means of the data in each group.}
\end{figure}

It is for the hardest, or noisy versions of bit retrieval that the hardness index provides the greatest utility. The Fourier magnitudes are then approximated by the noisy autocorrelations: $|\hat{s}_q|^2=\hat{a}_q/\sqrt{N}\approx \hat{n}_q/\sqrt{N}$. Again, the hardness index will vanish if any of the magnitude approximations is consistent with zero when noise is taken into account.

We define average-case instances by the property that their hardness index is near the median for random sequences of the same length. To generate these, we start with many random sequences, sort them by their indices, and extract a small sample at the center of the sorted list. In the limit of large $N$ the random variable $h(s)$ has a narrow distribution (appendix \ref{sec:harddist}) so that average-case instances are characterized by the average hardness:
\begin{equation}
\langle h\rangle_\delta\sim 4\delta(1-\delta) e^{-\gamma}.
\end{equation}
Here $\gamma\approx 0.577$ is Euler's constant and $\langle\cdots\rangle_\delta$ denotes the expectation value in a slightly generalized distribution, where the signs are still independent but might have a bias:
\begin{equation}
\langle s_k \rangle_\delta=1-2\delta,\qquad 0\le k\le N-1.
\end{equation}
The parameter $\delta$ corresponds to the density of $-1$'s in the sequence. For un\-biased signs the average-case hardness is $e^{-\gamma}\approx 0.561$. Not surprisingly, small $\delta$ corresponds to an easy limit of bit retrieval.

\section{Sparsity}\label{sec:sparsity}

When the density $\delta$ of $-1$'s in a sequence $s$ is small, the Fourier magnitude $|\hat{s}_0|^2$ is large and all others small; the resulting hardness $h(s)$ will then be low. A more direct way to see that bit retrieval in this limit is easy is to consider the polynomials $b(x)\in Z_N$ introduced in section \ref{sec:noise} with coefficients $b_k=(1-s_k)/2\in \{0,1\}$. These polynomials are sparse in the usual sense of having only a few terms. The autocorrelation $b(x)b(1/x)$ is now the input for bit retrieval and may also be sparse. Bit retrieval in the low density limit is therefore the problem of reconstructing a sparse polynomial from its sparse autocorrelation.

Our analysis of sparse bit retrieval is limited to the noise free case and uses the framework of cyclic difference sets. The \textit{difference set} $D$ is the set of powers in $b(x)$ while the \textit{group set} $G$ is the multiset of nonzero elements in $D-D$ that appear in $b(x)b(1/x)$. For example, if $N=5$ and
 \begin{equation}
 b(x)=1+x+x^2,
 \end{equation}
 then
 \begin{eqnarray}
 D&=&\{0,1,2\}\\
 G&=&\{1,1,2,3,4,4\}.
 \end{eqnarray}
 When $D$ has $K$ elements the number of elements in $G$ is $K(K-1)$. Cyclic difference sets are combinatorial designs where $D$'s are constructed such that $G$ contains all the elements $1,\ldots,N-1$ with equal multiplicity. This property is far from satisfied in the sparse limit $K\ll N$. Bit retrieval in the language of cyclic difference sets is the problem of reconstructing $D$ given $G$.

We will analyze a simple back-tracking tree search algorithm for reconstructing a difference set $D$ from a group set $G$. The depth $k$ in the tree corresponds to the number of elements of $D$ that have been proposed. When depth $k=K$ is reached, and the differences $D-D$ coincide with $G$, the algorithm terminates.

Without loss of generality we may take the first element of $D$ to be 0. Subsequent elements are proposed from the perspective of the 0 element of $D$. Specifically, each of the remaining $K-1$ elements of $D$ must at least be elements of $G$. In addition, for each proposed new element of $D$  the differences with the nonzero elements of $D$ must also be checked for membership in $G$. When a proposal at depth $k$ is successful, the newly added $k\mathrm{th}$ element of $D$ is said to claim $k-1$ elements of $G$ and make them unavailable for subsequent proposals at greater depth in the tree. This can be done efficiently by maintaining a set $G'$ of unclaimed elements of $G$. The elements in $G'$ are proposed in turn until a successful one is found whereupon the depth is incremented. If none of the elements in $G'$ produces a successful proposal, the depth is decremented, the proposal at that level declared unviable, and the next element of $G'$ (for that level) is considered.

By ordering the elements in $G$, and preserving that order in the unclaimed elements $G'$, set membership can be checked efficiently. The ordering of $G$ also determines the order in which the elements of $D$ are proposed. 

Our analysis of the runtime is in the limit of large $N$, and where the average multiplicity of the elements of $G$ is held constant. In the language of bit retrieval, this corresponds to the limit where the average coefficient of $b(x)b(1/x)$ stays constant (excepting the coefficient of $x^0$). We define the \textit{mean multiplicity} as
\begin{equation}
\mu=\frac{K(K-1)}{N-1}\sim \frac{K^2}{N}=\delta^2 N.
\end{equation}
By keeping $\mu$ constant in the large $N$ limit, the density decreases as $\delta\sim\sqrt{\mu/N}$ and so does the hardness, since $h\sim (4 e^{-\gamma})\delta$.

There are two regimes, depending on the frequency of the zero multiplicity elements (differences completely missing from $G$). We model\footnote{This is a model insofar as only very special $G$ correspond to any $D$.} the elements of a random $G$ as Poisson samples from the set $\{1,\ldots,N-1\}$. In this model, the probability of multiplicity 0 is $e^{-\mu}$. When an element of $G$ is proposed as an element of $D$ at depth $k$ of the tree, the probability of success is less than
\begin{equation}
(1-e^{-\mu})^{k-1},
\end{equation}
because the mean multiplicity $\mu$ of the unclaimed elements $G'$ decreases with the depth $k$. The simplest regime is the case $\mu\ll1$, where the success probability in our random model decays as $\mu^{k}$. This means that proposals are almost never successful by luck, but only because the proposed element belongs to the true difference set $D$. In this regime of $\mu$ there is almost never any backtracking: the algorithm simply runs through the elements of $G$ in turn, collecting those elements that belong to $D$ and skipping over the rest (which are eventually removed as differences to non-zero elements). The runtime has a $O(K^2)$ contribution from the proposals aborted almost always after a single $G$-membership check, and another $O(K^2)$ contribution, cumulatively, from checks that proved successful. Altogether then, the complexity is $O(\mu N)$ when $\mu\ll1$.

When $\mu\gg 1$ the search tree acquires width and backtracking contributes significantly to the complexity. Again using the Poisson model, an estimate of the number of successful proposals at depth $k$ (search tree nodes) is
\begin{equation}\label{n}
n(k)=\binom{N-1}{k-1}(1-e^{-\mu})^{\binom{k}{2}}.
\end{equation}
Apart from the 0 element at depth $k=1$, the algorithm will try all combinations of the nonzero elements because the probability any of them has zero multiplicity, $e^{-\mu}$, is very small. We restrict ourselves to depths $k\ll K$ so that our estimate \eqref{n} is still accurate when we account for the fact that proposals at lower depths have claimed a fraction $(k/K)^2$ of the elements of $G$. Applying Sterling's approximation and $\mu\gg 1$ to \eqref{n}, we find (details in appendix \ref{sec:maxwidth}) that the search tree has maximum width $n(k^*)$ for
\begin{equation}
k^*\sim e^\mu \log{N}, \qquad N\to\infty.
\end{equation}
This is subdominant to $\sqrt{\mu N}\sim K$, consistent with our assumption $k\ll K$.

Our runtime estimate for the case $\mu\gg 1$ is the work performed up to reaching depth $k^*$, where the tree has the most nodes. The algorithm will make
\begin{equation}\label{backtrack}
\binom{N-1}{k^*-1}\sim N^{e^\mu\log{N}}, \qquad N\to\infty
\end{equation}
proposals before it finds one of the $n(k^*)$ successful ones, and one of these in particular that has a branch that extends all the way to $k=K$. Since solutions (difference set reconstructions) are unique up to symmetry when there is no noise, the algorithm will make this number of proposals, reduced by symmetry, before a correct $k^*$-element subset of $D$ is discovered. As there are $2K$ solutions that have $0\in D$, the symmetry reduction is $O(\delta N)$ and subdominant relative to \eqref{backtrack}. There are  $(k^*-1)(k^*-2)/2$ checks of $G$-membership that go with each proposal, but since this is $O((\log{N})^2)$, its contribution to the complexity, multiplicatively, is subdominant. Contributions from work performed at greater depths are also subdominant. 

Our analysis shows that the complexity of difference set reconstruction by back-tracking tree search crosses over from $O(N)$, in problems where the mean multiplicity $\mu$ is fixed at small values, to the variable-exponent form \eqref{backtrack}, when $\mu$ is fixed at large values. Though the latter case still corresponds to ever sparser problems as $N$ grows, the exponential behavior of the exponent with $\mu$ makes the back-tracking algorithm impractical for even modest $\mu$. This behavior is in sharp contrast with the RRR algorithm, whose behavior with respect to hardness was already reported in section \ref{sec:hardness}. Figure 3 shows the behavior of the RRR iteration count with $N$ when $\mu$ is fixed at $2,4,6$ and $8$. Although there is exponential growth in the iteration count with $\mu$, the exponent of $N$ in all cases is small, perhaps even consistent with zero. Since the complexity of a single RRR iteration is $O(N\log{N})$, the RRR algorithm apparently maintains a low exponent complexity with $N$ in a sparseness regime (large $\mu$) where the back-tracking algorithm does not.

\begin{figure}[!t]
\begin{center}
\includegraphics[width=4.in]{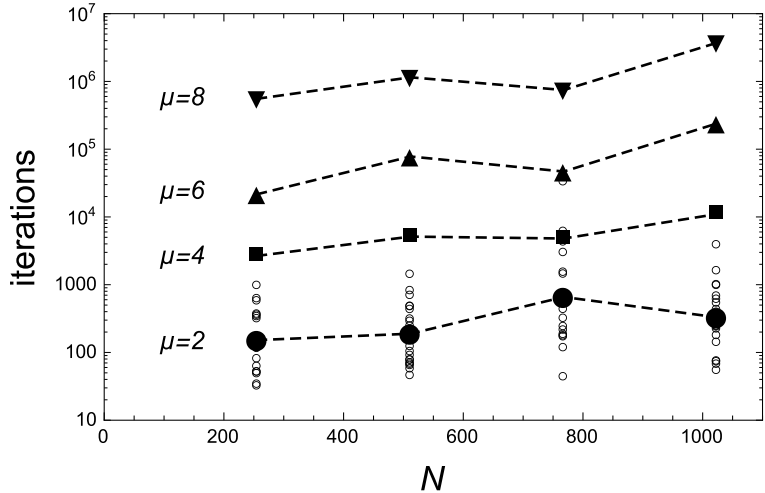}
\end{center}
\caption{Slow growth of the RRR iteration count with $N$ when the mean multiplicity $\mu$ is held fixed. The four plot symbols correspond to the indicated $\mu$ values. Each solid data point is the geometric mean of 20 random instances; empty circles show the scatter in the 20 instances for the series with $\mu=2$.}
\end{figure}

Phase retrieval in crystallography, to the extent that it can be modeled by bit retrieval, corresponds to the sparse limit. In a protein crystal the nitrogen, carbon and oxygen atoms scatter x-rays with similar strength and much more strongly than hydrogen. It is therefore not a bad approximation to model the contrast as equal 1 bits in a field of 0's. The number density of non-hydrogen atoms in a protein crystal is roughly $10^{-2}/\mathrm{\AA}^{3}$, and the best diffraction data can resolve the contrast to a scale of $1\,\mathrm{\AA}^3$ voxels. From these two numbers we infer $\delta\approx 10^{-2}$. However, as the analysis of the back-tracking algorithm has shown, bit retrieval can be hard even when $\delta$ is small. The more relevant parameter is the mean multiplicity, given by the product $\mu=\delta K$, where $K$ is the number of 1's (number of non-hydrogen atoms). It is interesting that the largest structures solved by strict phase retrieval, or ``direct methods" \cite{US}, have $K\approx 10^3$ and therefore $\mu\approx 10$.

\section{Convex relaxation}\label{sec:convex}

In the presence of noise, and without sparsity, we do not even have a sub-expo\-nential-time algorithm for bit retrieval. The challenge is then to find the algorithm that minimizes $c$ in the exponential complexity $2^{cN}$. While always bleak in practical terms, from a theoretical perspective an exponential complexity has the redeeming feature that we can be cavalier about the implementation, as those details usually contribute only a polynomially dependent factor.

As a first attempt to bound $c$ away from 1 we can ask: How small of a fraction of the signs, upon being given guessed values, allows for an easy determination of the extensibility of the guess into a complete solution? Here ``easy" refers to any polynomial-time computation. Perhaps the simplest algorithm of this kind is based on the observation, that if the first $k$ signs are guessed and only the remaining $N-k$ are treated as unknowns, then for suitable $k$ there will be at least $N-k$ of the quadratic autocorrelation equations \eqref{auto} that are linear in the unknowns. By simple counting we arrive at the sufficient condition
\begin{equation}\label{simple}
k\ge N-\lfloor N/2\rfloor/2.
\end{equation}
In the absence of noise, the $N-k$ linear equations for $N-k$ real variables can be solved by Gaussian elimination or, in the presence of noise, the linear inequalities are solved by linear programming. We are done if there is a solution/feasible point where all the variables are $\pm 1$ and also satisfy the remaining quadratic equations/inequalities; if not, we make another guess for the first $k$ signs and repeat. This analysis fails for the class of instances where every $k$-subsequence of the solution signs gives a singular matrix for the linear system. Assuming our instance is not in this class, then trying all $2^k$ guesses gives an algorithm with $c=3/4$ by \eqref{simple}. This is surely a very poor bound on $c$ since the method uses only half of the available autocorrelation data\footnote{Curiously, when $N$ is divisible by 3 the bound drops to $c=2/3$. Rather than guess the first $k$ signs, in this case we guess the signs at all $k\equiv \pm 1 \pmod{3}$ and solve $N/3$ equations/inequalities that are linear in the remainder.}.

We can get a smaller constant $c$ with an algorithm that uses all the autocorrelation data, as well as bounds on the variables. Starting from the constraint sets $A(a)$ and $B$ introduced in section \ref{sec:hardness}, the improved algorithm follows from the observation that all $x\in B$ have the same 2-norm and therefore the autocorrelation constraint set may be made convex:
\begin{equation}\label{Aconvex}
\overline{A}(a)=\left\{x\in \mathbb{R}^N \colon |\hat{x}_q|^2 \le \hat{a}_q/\sqrt{N},\quad 0\le q\le \lfloor N/2\rfloor \right\}.
\end{equation}
The statement $x\in \overline{A}(a)\cap B$ implies $x\in A(a)\cap B$ because all the inequalities in \eqref{Aconvex} have to be saturated in order that $x$ has the required 2-norm. Bit retrieval remains hard because $B$ is still non-convex.

The feasibility of a partial assignment of $\pm 1$ values to the variables in the constraint satisfaction problem can now be tested by solving a convex problem. Choosing to assign values in consecutive order, we define the following convex relaxations (facets) of the hypercube constraint set:
\begin{equation}\label{facet}
B(s_1,\ldots,s_K)=\left\{x\in \mathbb{R}^N \colon 
\begin{array}{ll}
x_k=s_k,& k\in\{1,\ldots,K\}\\
 |x_k|\le 1,&k\notin\{1,\ldots,K\}
 \end{array}
 \right\}.
\end{equation}
If we find that the convex set $\overline{A}(a)\cap B(s_1,\ldots,s_K)$ is empty, then we know that the solution does not have consecutive signs $s_1,\ldots,s_K$. The tree of sign assignments is searched exactly as in the branch and bound algorithm for integer programming. Whenever the set intersection is not empty, $K$ is incremented; otherwise, $x_K$ is assigned the other sign, or if that has already been tried, $K$ is decremented. 

The problem to be solved at each node of the branch and bound tree is finding a point in the intersection of two convex sets or producing a proof that such a point does not exist. This is solved by the ellipsoid method \cite{K} in polynomial time when we are provided with two things: (i) a lower bound on the volume of any feasible region, and (ii) a polynomial-time \textit{separation oracle}. The first condition simply excuses us from failing to find feasible points when the volume of the feasible region is too small. We should therefore only use the proposed branch and bound algorithm for fixed-precision bit retrieval, where the noise parameter $\eta$ gives us license to very slightly weaken the inequalities in \eqref{Aconvex} as well as thicken the hypercube facets \eqref{facet} so as to make them full-dimensional. These refinements give the volume of the feasible region a lower bound.

A separation oracle for convex set $C$, when given a point $x$, either declares \mbox{$x\in C$} or returns a hyperplane that separates $x$ from $C$. Such an oracle can be implemented in polynomial time when, as in our problem, $C=A\cap B$ is the intersection of two convex sets and \textit{projections} $P_A$ and $P_B$ to these sets can be computed in polynomial time. The projection of point $x$ to $A$ (and analogously for $B$) is a point $P_A(x)\in A$ that minimizes the 2-norm to $x$ up to a bound set by the precision\footnote{The computed projection is within a ball of radius set by $\eta$ of a true distance minimizing point in $A$. Note that in the finite precision context $A$ is a finite set and one makes no distinction between the closed and open topology.}.

The implementation of the separation oracle with projections encounters two cases. After computing $p_A=P_A(x)$ and $p_B=P_B(x)$, either $x=p_A=p_B$ and we know $x\in C$, or one of the projections, say $p_A$, is distinct from $x$. The required hyperplane, in the second case, is the co-dimension-1 hyperplane that passes through $p_A$ and is orthogonal to the line passing through $x$ and $p_A$. That this is a valid separating hyperplane is explained in the caption to Figure 4.

\begin{figure}[!t]
\begin{center}
\includegraphics[width=3.in]{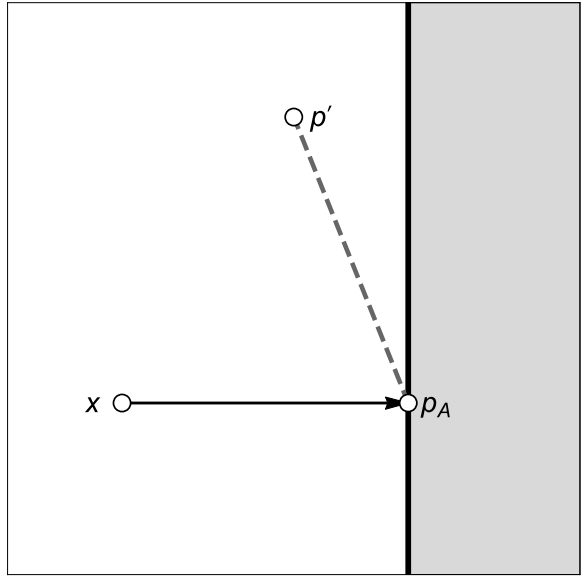}
\end{center}
\caption{Separating hyperplane construction by projections. The point $x$ has been projected to $p_A=P_A(x)$, the nearest point on convex set $A$. The proposed hyperplane is shown passing through $p_A$ and orthogonal to the line between $p_A$ and $x$, with the feasible region that contains $C=A\cap B$ shown shaded. If there existed a point $p'\in C$ not in the region defined by the hyperplane, then because $C$ is convex, all the convex combinations of $p_A$ and $p'$ (dashed line) would also belong to $C$. But that is impossible because then there would be a point in $C\subset A$ closer to $x$ than $p_A$ (in the convex combination and close to $p_A$).}
\end{figure}

The projections to $\overline{A}(a)$ and $B(s_1,\ldots,s_K)$ can be computed in  $O(N\log{N})$ and $O(N)$ time, respectively. To project to $\overline{A}(a)$ we note that since the Fourier transform preserves the 2-norm, we may work with the inequalities in \eqref{Aconvex} directly. The 2-norm minimizing projection map either leaves $\hat{x}_q$ unchanged, if the magnitude inequality is satisfied, or replaces the magnitude of $\hat{x}_q$ by the magnitude that saturates the bound (leaving the phase unchanged). The complexity of the computation is dominated by the fast Fourier transform, first from $x$ to $\hat{x}$, and then back to $x$. Projecting to the facet \eqref{facet} involves either replacing $x_k$ by $s_k$ when $k\in\{1,\ldots,K\}$, or when $k\notin\{1,\ldots,K\}$, leaving $x_k$ unchanged when its magnitude does not exceed 1, or replacing it by $\pm 1$, whichever has the same sign. Both projections are easily modified, with almost no additional effort, when the constraint sets are replaced by their noisy counterparts.

Because we have a lower bound on the volume of the feasible set and a separation oracle that can be implemented in polynomial time, the convex feasibility problems we need to solve at the nodes of our branch and bound tree can be solved in polynomial time. The complexity is therefore dominated by the exponential growth in the number of nodes. Finding good bounds on the number of nodes at depth $K$, $n(K)$, is a difficult problem for a general instance with autocorrelation $a$. However, using sampling techniques we have been able to learn much about $n(K)$ and estimate the constant $c$ of the branch and bound algorithm.

We can express the number of nodes at depth $K$ as
\begin{equation}
n(K)=2^K p(K|a),
\end{equation}
where $p(K|a)$ is the probability, conditional on the autocorrelation $a$ of the constraint \eqref{Aconvex}, that a randomly selected facet \eqref{facet} produces a feasible point in the convex constraint problem. By sampling many random facets of the hypercube (length-$K$ sign sequences) this probability can be determined to sufficient precision that we get a good estimate of $n(K)$.
Anticipating exponential growth, we define the scaled logarithmic tree-width:
\begin{equation}
w=\frac{1}{N}\log_2{n(K)}=K/N+\frac{1}{N}\log_2{p(K|a)}.
\end{equation}
For small $K$ almost all sign sequences produce a feasible point, $n(K)\sim 2^K$, and $w\sim K/N$. When $K$ is near the upper end of its range, $w$ will return to zero in the case of low noise, since solutions and $n(N)$ have vanishing entropy. This behavior is displayed in Figure 5 for five $N=48$ average-case instances we sampled. The log-width $w$ is plotted against the fractional depth $y=K/N$ because we anticipate that $w(y)$ will have a similar shape for different sizes $N$, at least when comparing instances of the same hardness.

\begin{figure}[!t]
\begin{center}
\includegraphics[width=4.in]{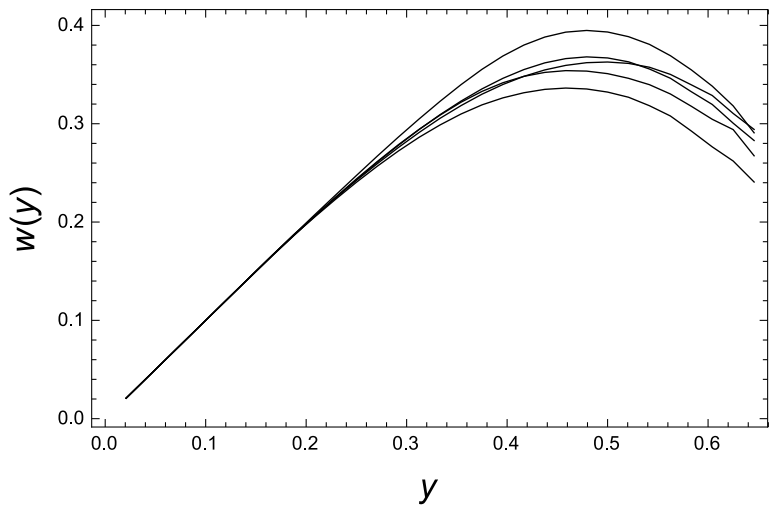}
\end{center}
\caption{Scaled log-tree-width $w(y)$ versus fractional tree depth $y$, of the branch and bound bit retrieval algorithm for five average-case $N=48$ instances.}
\end{figure}

To determine the exponential complexity we note that, due to the symmetry of solutions, on average a $O(1/N)$ fraction of the nodes at the fractional depth $y^*=K^*/N$ of maximum log-width $w(y^*)$ have to be explored before a branch is found that extends to $K=N$. Including the work performed at other depths only contributes a polynomial factor (the total depth is $N$), as does the polynomial-time work performed at each node. Since $n(K^*)=2^{w(y^*)N}$, we estimate the complexity constant as
\begin{equation}\label{wmax}
c=\max_y w(y).
\end{equation}
Although there is considerable variation from one instance to another, we see from Figure 5, that at least for average-case $N=48$, the tree has greatest width near $y=0.5$ and $c\approx 0.36$.

We obtain much better evidence of $w(y)$ converging to a large-$N$ limit (as in statistical mechanics) when we consider special families of instances. We chose the perfect Hadamard sequences as these maximize the hardness index. The constant \eqref{wmax} derived from this family should be a strong candidate for the worst-case complexity of the branch and bound algorithm. Figure 6 shows the log-widths for perfect Hadamard instances of size $N=20, 40, 80$ and $160$. The noise free, $k\ne 0$ autocorrelations for all these instances is $a_k=-1$. Though solutions in the absence of noise exist only for particular odd $N$, the small noise required to have solutions at the chosen $N$ should have negligible effect on the complexity of this algorithm, which is determined by the widest part of the tree. Since we are mostly interested in the complexity bound, our samples were confined to depths where the search trees have their greatest width.

The maxima of the four $w(y)$ curves for prefect Hadamard instances have a systematic behavior\footnote{We have no reason to suspect that those values of $N$ for which perfect Hadamard sequences actually exist would deviate from the observed behavior.} with $N$ and enable us to extrapolate, in Figure 7, both the fractional depth of the maximum and the value at the maximum to their $N=\infty$ values. We find $y_\infty\approx 0.630$ and $w(y_\infty)\approx 0.564$. Our bound for the exponential complexity constant is therefore $c<0.564$. Although much worse than average-case, this is still significantly better than the bound $c<3/4$ obtained at the start of this section, where the tree is searched exhaustively to fractional depth $y=3/4$ and the remaining variables are found by solving linear equations.

\begin{figure}[!t]
\begin{center}
\includegraphics[width=4.in]{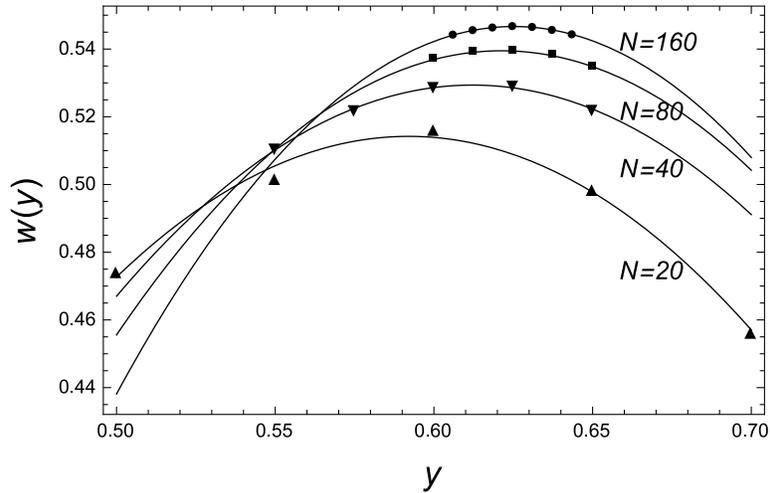}
\end{center}
\caption{Behavior of the scaled log-tree-width $w(y)$ of the branch and bound bit retrieval algorithm as the sizes of Hadamard instances are doubled. Both the location and value of the maximum appear to be converging. Extrapolations of the location $y^*$ and value $w(y^*)$ from these data are shown in Figure 7.}
\end{figure}

\begin{figure}[!t]
\begin{center}
\includegraphics[width=4.5in]{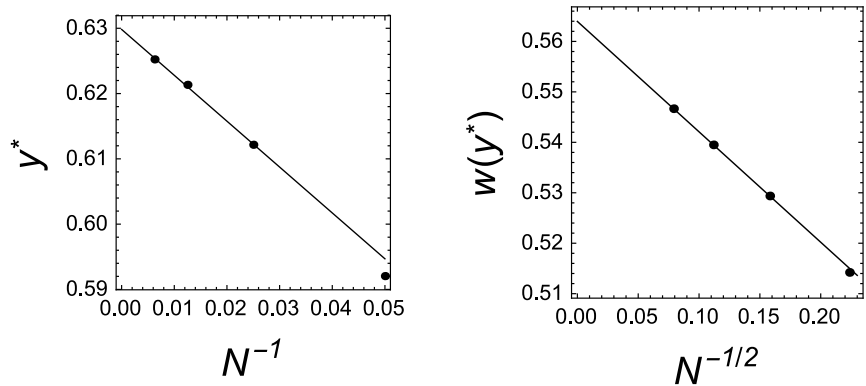}
\end{center}
\caption{Extrapolations of the data in Figure 6, assuming the finite-$N$ corrections to $y^*$ and $w(y^*)$ are respectively proportional to $N^{-1}$ and $N^{-1/2}$.}
\end{figure}

\section{Relaxed-reflect-reflect algorithm}\label{sec:RRR}

The RRR algorithm has an element of randomness making the run-time unpredictable. Figure 8 plots the distribution of run-times (iterations) in $2\times 10^5$ solutions of the $N=43$ Hadamard instance. The most probable runtime is near zero and the distribution decays exponentially, for a mean run-time of $1.7\times 10^5$ iterations. These statistics are consistent with an algorithm that blindly and repeatedly reaches into an urn of $M$ solution candidates, terminating when it has retrieved one of the $4\times 43$ solutions. Two questions immediately come to mind. The easier of these is: How can an algorithm that is deterministic over most of its run-time behave randomly? The much harder question is: How did the $M=2^{43}$ solution candidates get reduced, apparently, to only about $1.7\times 10^5\times (4\times 43)\approx 2^{25}$?

\begin{figure}[!t]
\begin{center}
\includegraphics[width=4.5in]{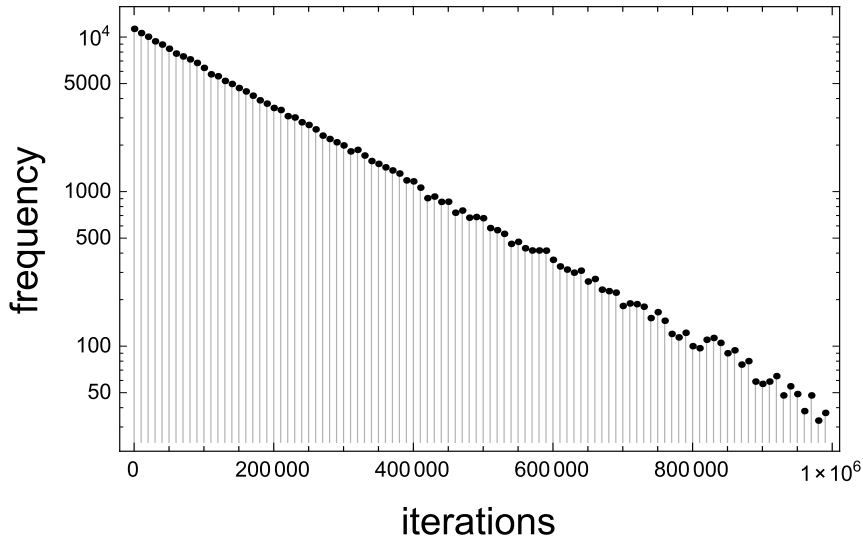}
\end{center}
\caption{Exponential distribution of run-times (iterations) of the RRR algorithm in $2\times 10^5$ solutions of the $N=43$ Hadamard instance.}
\end{figure}

A very simple strategy for expanding the reach of naive sampling uses the projections to the two constraint sets introduced in section \ref{sec:hardness}, $A(a)$ and $B$. Projection $P_{A(a)}(x)$ is the sequence having smallest 2-norm distance to $x$ and autocorrelation $a$, and $P_B(x)$ is the $\pm 1$ sequence obtained from $x$ by rounding. Suppose we start with a sequence of signs $x_0$. If $P_{A(a)}(x_0)=x_0$ we are done, because our guessed signs have the correct autocorrelation; otherwise, we construct the nearest sign sequence: $x_1=P_B(P_{A(a)}(x_0))$. We continue doing this until we find a solution or, more probably, a fixed point $x_{i+1}=P_B(P_{A(a)}(x_i))=x_i$ that is not a solution ($P_{A(a)}(x_i)\ne x_i$). From experiments we find that non-solution fixed points are encountered after just a few iterations and differ from the initial guess by a  Hamming distance of only about $N/30$, for Hadamard instances of size $N$. Although this method is an improvement over naive sampling, the fraction of signs that the projections modify is too small to make a useful algorithm.

The RRR algorithm is a far superior method of generating samples. Originally proposed (without relaxation) by Bauschke, Combettes and Luke\footnote{Although Douglas and Rachford are usually given credit for the first application, \cite{BCL} were the first to notice, more generally, the averaged-alternating-reflection structure of the iteration.} \cite{BCL} as a method for finding two points that achieve the minimum distance between two convex sets, much of the analysis that applies in that setting has little relevance for the highly non-convex constraints of bit retrieval. In particular, it is no longer useful to study convergence, or the notion that the algorithm makes systematic progress towards solutions. A potentially more productive goal, given the circumstances, is to discover general principles for iteratively constructing good samples, and in particular, understanding how this is achieved by the RRR algorithm. This is the approach that we will take.

The algorithm is usually initialized with some $x\in \mathbb{R}^N$ produced by a random number generator. This is the only explicit use of randomness and provides a means for exploring statistical properties, such as the run-time distribution in Figure 8. There are no initial $x$ that are inherently better or worse than others; in fact, the initial $x$ may even be a sequence of signs. After selecting the initial $x$, the following map is applied iteratively:
\begin{equation}\label{RRR1}
x\mapsto x'=x+\beta\left(P_{A(a)}(2 P_B(x)-x)-P_B(x)\right).
\end{equation}
Here $\beta$ is a real parameter with the restriction $0<\beta<2$ as explained below. The samples of interest to bit retrieval are $P_B(x)$, not $x$. Unlike the earlier scheme for generating samples, here we find that fixed points are always associated with solutions. Suppose $x^*$ is a fixed point; then, since
\begin{equation}
0=P_{A(a)}(2 P_B(x^*)-x^*)-P_B(x^*),
\end{equation}
we see that $P_B(x^*)$ is in the range of $P_{A(a)}$, indicating it is a sequence having the required autocorrelation.

For the fixed points of the RRR map to be the basis of an algorithm, it is necessary that these are attractive. To analyze this local property we consider a $p_B\in B$ that is either a solution, so $p_B\in A(a)\cap B$, or a near solution. A near solution $p_B\in B$ has the property that the projection $p_A=P_{A(a)}(p_B)$ is very close to $p_B$, as measured by the 2-norm. In either case, we can study the behavior of the map for $x$ that are near to both $p_A$ and $p_B$. For such $x$, $P_B(x)=p_B$, since all other elements of the hypercube $B$ are more distant, and $A(a)$ may be approximated by its tangent space, an affine space whose proximal point to $p_B$ is $p_A$. The dimension of the affine space is $M=\lfloor(N-1)/2\rfloor$, the number of circles in the Cartesian product description of $A(a)$.

For the local analysis of the near solution case we use $p_B$ as the origin and the orthogonal decomposition $\mathbb{R}^N=\mathbb{R}^1\oplus\mathbb{R}^M\oplus\mathbb{R}^{N-M-1}$, where the first component is parallel to $p_A-p_B$, the second component is the linear space (approximation of) $A(a)-p_A$, and the last component is the orthogonal complement of these. In this decomposition a general point is written as
\begin{equation}
x=x_1\oplus x_A\oplus x_\perp,
\end{equation}
and the two constraint sets have the form
\begin{eqnarray}
A(a)&=&d_A\oplus \mathbb{R}^M\oplus 0\\
p_B&=&0\oplus 0\oplus 0,
\end{eqnarray}
where $d_A=\|p_A-p_B\|$. Using the following formulas for general (local) projections,
\begin{eqnarray}
P_{A(a)}(x)&=&d_A\oplus x_A\oplus 0\\
P_{B}(x)&=&0\oplus 0\oplus 0,
\end{eqnarray}
we obtain the result of one iteration of the RRR map \eqref{RRR1}:
\begin{equation}\label{approxRRR}
x'=(x_1+\beta d_A)\oplus (1-\beta)x_A\oplus x_\perp.
\end{equation}
This formula is valid also for the true solution case, $d_A=0$, the only difference being that there now is no longer a distinction between the $\mathbb{R}^1$ and $\mathbb{R}^{N-M-1}$ components of the orthogonal decomposition.

When $d_A=0$, we see from \eqref{approxRRR} that we have the stable fixed points
\begin{equation}
x^*=x_1\oplus 0 \oplus x_\perp
\end{equation}
if and only if $|1-\beta|<1$, the condition asserted earlier. A more precise statement is that we have a $N-M$ dimensional space of fixed points. All points in this space produce the same bit retrieval sample, $p_B=P_B(x^*)$, a solution. In the case of a near solution, $d_A> 0$, although there is no longer a fixed point, the RRR map is still contracting in the $M$-dimensional tangent space approximation of $A(a)$. While it is the unidirectional motion in the first component, $x'_1=x_1+\beta d_A$ (purposeful escape from a non-solution), that is usually credited for the algorithm's success \cite{ETR}, the contracting behavior that goes with it may be just as significant.

The algorithm's name derives from the fact that the map \eqref{RRR1} can be written more compactly,
\begin{equation}\label{RRR2}
x\mapsto x'=(1-\gamma)x+\gamma\, R_{A(a)}\circ R_B(x),
\end{equation}
in terms of the reflections
\begin{eqnarray}
R_{A(a)}(x)&=&2P_{A(a)}(x)-x\\
R_{B}(x)&=&2P_{B}(x)-x,
\end{eqnarray}
and where $\gamma=\beta/2$ looks like a relaxation parameter. Figure 9 shows how the combination of two reflections followed by the $\gamma$-average has the effect of contracting along the tangent space of the $A(a)$ constraint. The case $\gamma=1/2$ is called AAR for averaged-alternating-reflections \cite{BCL}. We argue later in this section that, in combinatorially hard feasibility problems such as bit retrieval, it is important to keep $\gamma$ small.

\begin{figure}[!t]
\begin{center}
\includegraphics[width=3.in]{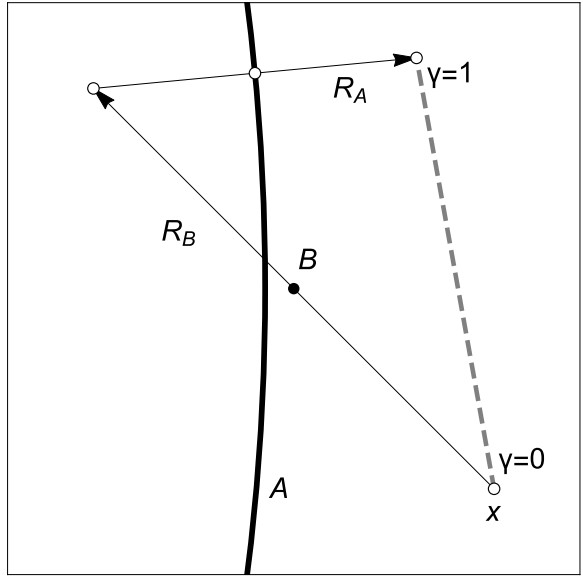}
\end{center}
\caption{Cartoon showing a single application of the RRR map to the point $x$. The result after two reflections, first $R_B$ and then $R_A$, is averaged with $x$ according to the value of the parameter $\gamma$ (a point along the dashed line). When the distance between $A$ and $B$ is smaller than the distance of either set to $x$, the chief effect of the map is to move $x$ parallel to the tangent space of $A$.}
\end{figure}

Part of the attractiveness of the RRR algorithm is the ease of its implementation. The computations involve almost exclusively floating point numbers. Most of the work in one iteration of \eqref{RRR1} is the pair of FFT's needed to perform the rescaling of the Fourier magnitudes to their known values. Even with single precision floating point numbers, so that $\eta$ in \eqref{finiteprecision} is of order $10^{-7}$, we do not compromise solution uniqueness until $N$ has grown as large as $\eta^{-2}$. For all practical purposes then, we can use a floating point RRR implementation to solve noise-free bit retrieval.

Since solutions correspond to fixed points of the RRR map, termination is linked to the value of $\|x'-x\|$. To eliminate any doubt that this floating point number has reached a small enough value, we can keep track of the smallest-achieved value over the course of a run and, whenever there is an improvement, compute the autocorrelation of $P_B(x)$ using integer arithmetic for a foolproof termination check.

One concern when using a map to generate samples, when its global behavior is complex, is that the iterates might converge on an unproductive cycle. This almost never happens for large $N$; for the $N=23$ Hadamard instance the probability of convergence to a cycle is already less than $10^{-3}$. Since both $P_{A(a)}$ and $P_B$ have strongly branching behavior when certain numbers (complex and real, respectively) are small in magnitude, it is no mystery why stable cycles are rare. Strongly mixing dynamics is also the best interpretation of the exponential run-time distribution (Fig. 8) with which we introduced this section. The latter can even be used to defend a simple safeguard against cycles: frequent random restarts. However, we did not implement this policy for the RRR results presented here.

\begin{figure}[!t]
\begin{center}
\includegraphics[width=4.5in]{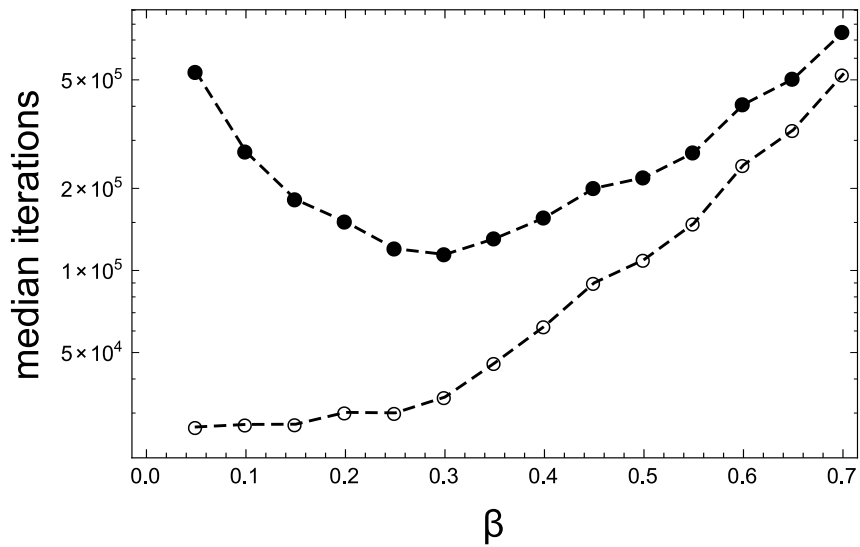}
\end{center}
\caption{\textit{Solid circles}: Median number of iterations, out of 1000 trials, required by the RRR algorithm to solve the $N=43$ Hadamard instance over a range of $\beta$ values. \textit{Open circles}: Same data but with the inverse time step, $\beta^{-1}$, divided out.}
\end{figure}

Although it is certainly possible that the optimal value of $\beta$ depends on $N$ and the hardness index, we sought a single value for all the experiments performed in this study. We determined this value by solving the $N=43$ Hadamard instance 1000 times from random starts over a range of $\beta$ values. The median iteration counts, plotted in Figure 10, have a broad minimum near $\beta=0.3$. The upturn at small $\beta$ is explained by the fact that the $\beta\to 0$ limit of \eqref{RRR1} is a system of differential equations, with $\beta$ as the time step. In this limit most of the time is spent between branch points, where the trajectory is only weakly affected by $\beta$. Because the branch points continue to scramble the trajectories for arbitrarily small $\beta$, the  character of the trajectory will not change in the $\beta\to 0$ limit. Since the probability of stumbling upon a solution (per branch point encountered along the way) is constant, the median iteration count per solution should scale in proportion to the number of steps between branch points, $\beta^{-1}$. This explanation is supported by the data in Figure 10, where we have factored out this time-step dependence and see that the result is independent of $\beta$ in the small $\beta$ limit. Because it looks like the continuous time limit of \eqref{RRR1} is well approximated already with time-step $\beta=0.3$, we have used this value in all of our experiments unless stated otherwise.

For RRR to be a proper algorithm for combinatorial search, we should at the very least have a model of the space wherein the search takes place. Developing such a model is the subject of ongoing research, and we can only offer some interesting, possibly relevant observations. We take as our primary clue the upturn of the median iteration count with increasing $\beta$, shown in Figure 10. It appears that RRR performs best in the \textit{flow limit}, $\beta\to 0$.

We have direct information about the $\beta\to 0$ search space from the statistics of the magnitudes of the individual components of $x$. Recall that the iterates $x$ of the RRR algorithm are rather indirectly linked to the constraint sets. A small magnitude of component $x_k$, for instance, corresponds to a strong uncertainty in the  sign $s_k$ of the projection $P_B(x)$. Figure 11 shows the distribution $|x|$ over all components taken from a single run of the algorithm on an instance with the Hadamard autocorrelation $a_k=-1$, $k\ne 0$, but with $N=41$ so there is no solution. We believe RRR is ergodic on these kinds of insoluble instances, and use a single long run to sample statistical data. To be sure to see flow limit behavior, we have set $\beta=0.01$. There is a clear anomaly in the distribution at $|x|=0$ whose width (detail in right panel of Fig. 11) scales with $\beta$. On average about six of the components of $x$ are exceptionally small in magnitude, the remainder being broadly distributed. Since zeroes in the components of $x$ correspond to the Voronoi cell faces of the hypercube $B$, we conclude that the search (in the flow limit) is confined to Voronoi cell facets whose codimension is about six in this problem instance.

\begin{figure}[!t]
\begin{center}
\includegraphics[width=5.in]{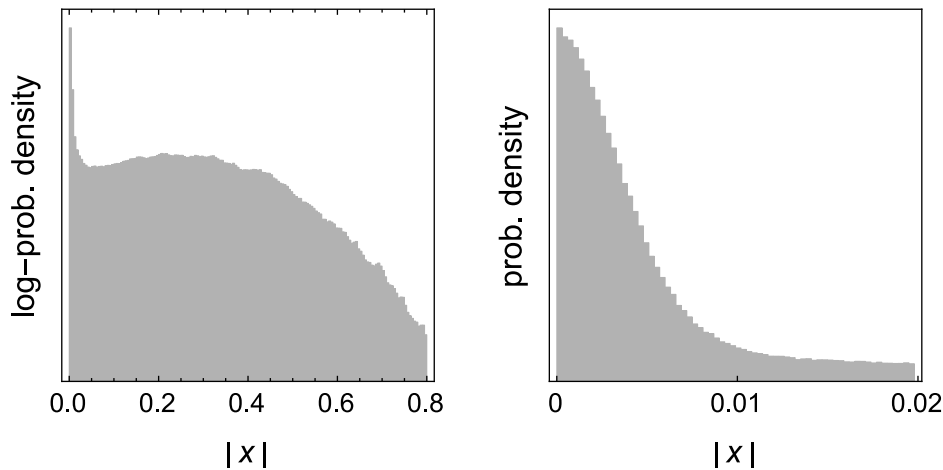}
\end{center}
\caption{Distribution of the magnitudes of the components of $x$ (detail on right) in the flow limit ($\beta=0.01$). Data is taken from an insoluble $n=41$ Hadamard instance. As the weight in the anomaly at $|x|=0$ is about 15\%, or 6 components of $x$, the RRR search is confined to codimension-6 Voronoi facets of the discrete set $B$.}
\end{figure}

The \textit{facet attraction property} exhibited by the data in Figure 11 raises two questions about the RRR algorithm: (1) What is the mechanism for this attraction, and (2) is this property responsible for the algorithm's good performance? The answer to the first question can be explained with the flow field in Figure 12. This example was constructed to have the codimension-2 attractor (Voronoi-facet) $x_1=x_2=0$. Near the origin of the $(x_1,x_2)$ plane, the point $p_B=P_B(x)$ jumps discontinuously between the four possibilities $\{-1,1\}^2$ and this in turn determines four discontinuous points $p_A=P_A(2p_B-x)$. The RRR flow field is parallel to $p_A-p_B$, and near the origin has the discontinuous structure shown. For finite $\beta$ the RRR iterates make finite jumps but stay within a distance of scale $\beta$ from the origin. We do not know what determines the codimension of the facets that RRR is attracted to. The codimension for Hadamard instances appears to be about $0.15 N$.

For the second question above we need to establish that the facet attraction property has the effect of narrowing the search to a smaller domain that has an increased rate of finding solutions. A possibly relevant statistic in this regard is a clustering property of the projections $p_A$. Suppose we have an instance of bit retrieval with solution $p_B$. Not only do we know that $p_B\in A$, but perturbations of $p_B$ will at least be close to points $p_A\in A$. If we now consider a set of perturbations of $p_B$, all with a nearby point $p_A$, then by the triangle inequality we will have a set of points $p_A\in A$ all within some bounded distance of each other. The existence of such a cluster in the set $A$ comes from the fact that the original point $p_B$ was a solution.

\begin{figure}[!t]
\begin{center}
\includegraphics[width=3.in]{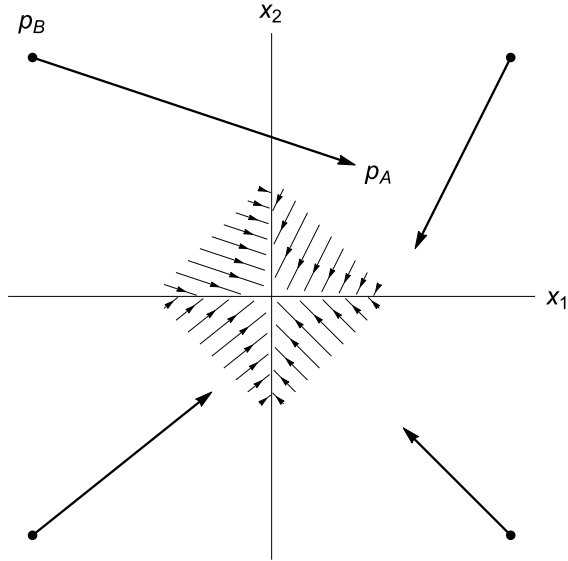}
\end{center}
\caption{Facet attraction is a simple consequence of particular flow fields ($\beta\to 0$ limit of RRR), shown here for a codimension-2 facet. In each orthant of the space orthogonal to the facet (just four in this example), the flow field (shown projected into the plane) near the origin is parallel to $p_A-p_B$, where $p_B=P_B(x)$ is the vector of signs that goes with each orthant and $p_A=P_A(2p_B-x)$. The origin is a fixed point for the flow shown.}
\end{figure}

The RRR facet attraction property is our motivation for the following perturbation of a solution $p_B$. Let $J$ be the index set of the components of $x$ that vanish on a facet whose codimension is $|J|$. Define $x_J$ as the point obtained from $p_B$ by setting to zero all components with indices in $J$. The projection $P_B(x_J)$ is undefined on the zero components but arbitrarily small perturbations $\tilde{x}_J$, when projected $P_B(\tilde{x}_J)$, produce a hypercube of $2^{|J|}$ points in the space orthogonal to the facet. The argument of the projection $P_A$ in RRR is $2 P_B(\tilde{x}_J)-\tilde{x}_J$, a vector whose components match the solution $p_B$ on the indices not in $J$ and includes all combinations of $\pm 2$ in the rest. This is our special set of perturbations of $p_B$ and as argued earlier, when projected by $P_A$ the resulting set of points will cluster when $p_B$ is a solution. What remains is to argue the converse: that the facet attraction property promotes clustering and thereby increases the odds that a $p_B$ generated by RRR is a solution.

The relationship between facet attraction and clustering of the projections $p_A$ (of the perturbations of $p_B$) is also explained in Figure 12. Note that there are constraints on the vectors $p_A-p_B$ such that the corresponding flow field is attracting for the origin. These vectors must all lie within solid angles subtended at vertices of the hypercube $\{-1,1\}^{|J|}$, constraints that are particularly strong in high dimensions (large $|J|$). Together with mild assumptions on the magnitudes $\|p_A-p_B\|$, the effect of these constraints is to bring the $p_A$ into proximity of each other (clustering).

The $N=41$ Hadamard instance we used in our demonstration of facet attraction is also well suited to demonstrate the degree of clustering of the points $p_A=P_A(x')$ obtained by projecting perturbations $x'$ of a particular $p_B$. As above, we consider the set of $2^{|J|}$ perturbations $x'$ specified by index set $J$, where $x'$ matches $p_B$ for indices not in $J$ and has values $\pm 2$ on the others. We use the root-mean-square measure of clustering
\begin{equation}
\sigma_A^2=\langle\, \|p_A-\langle{p_A}\rangle \|^2\,\rangle,
\end{equation}
where the angle brackets are averages over the $2^{|J|}$ projections $p_A$. Figure 13 shows distributions of $\sigma_A$ for three choices of the base point $p_B$ being perturbed. The distribution with the smallest mean, not surprisingly, is generated by solution points $p_B$. To produce this distribution we used the two-valued symmetric-Hadamard autocorrelations $a_k\in\{-3,1\}$, $k\ne 0$ rather than $a_k=-1$ (for which there are no solutions). The $\sigma_A$ values have a distribution because we uniformly sample the index sets $J$, for $|J|=6$.

For the other two distributions in Figure 13 we used the insoluble autocorrelation data $a_k=-1$, $k\ne 0$. The distribution with the largest mean was generated by uniformly sampling $p_B\in B$. That the mean for this distribution is higher than the distribution for a solution point is of course not surprising; the separation of the distributions just establishes the scale of the clustering effect. The most interesting distribution is the middle one, for $p_B$ samples generated by iterating RRR (with $\beta=0.01$). We see that RRR has the desired effect of generating samples $p_B$ with better clustering, or solution likelihood, than random samples. Moreover, the mechanism for the improved clustering is linked to the facet attraction property (Fig. 12).

\begin{figure}[!t]
\begin{center}
\includegraphics[width=4.in]{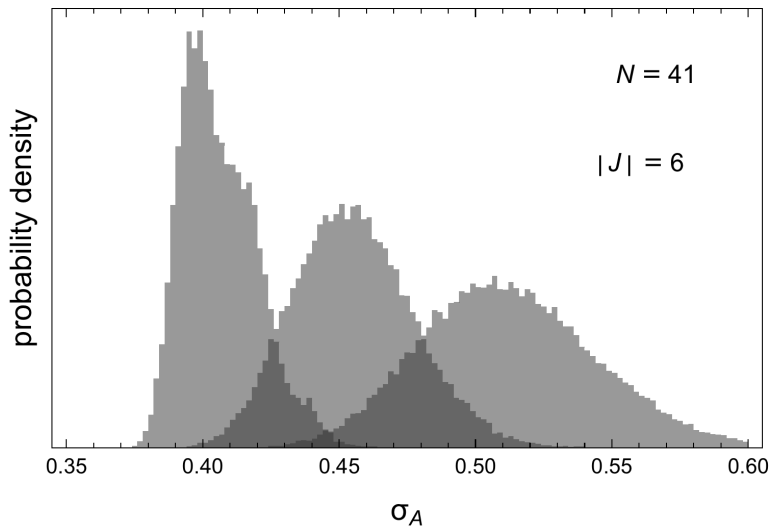}
\end{center}
\caption{Cluster size ($\sigma_A$) distributions on constraint set $A$ generated by perturbations of three types of points $p_B\in B$: solutions (lowest mean), random samples (highest mean), samples generated by RRR (intermediate mean). Data are for the $N=41$ Hadamard instance and codimension-6 facets.}
\end{figure}

The complexity of the RRR algorithm for bit retrieval can be assessed by the exponential growth in the number of iterations, since the work in each iteration has only $O(N\log{N})$ growth from the pair of FFTs in the $P_{A(a)}$ projection. Our experiments support the growth law, $e^{c N}$, where $c$ is hardness dependent. 
We report values of the median of the iteration counts obtained in 20 trials, since we have seen no exceptions to the exponential form of the iteration distribution on individual instances (Fig. 8).
Figure 14 shows the exponential growth in the median for average-case instances ($h\approx 0.56$). Results are shown for ten instances at each $N$, as well as the geometric means of the ten instances. While there is considerable scatter with respect to instance, the averages are consistent with an exponential growth law and $c\approx 0.212$. 

\begin{figure}[!t]
\begin{center}
\includegraphics[width=4.5in]{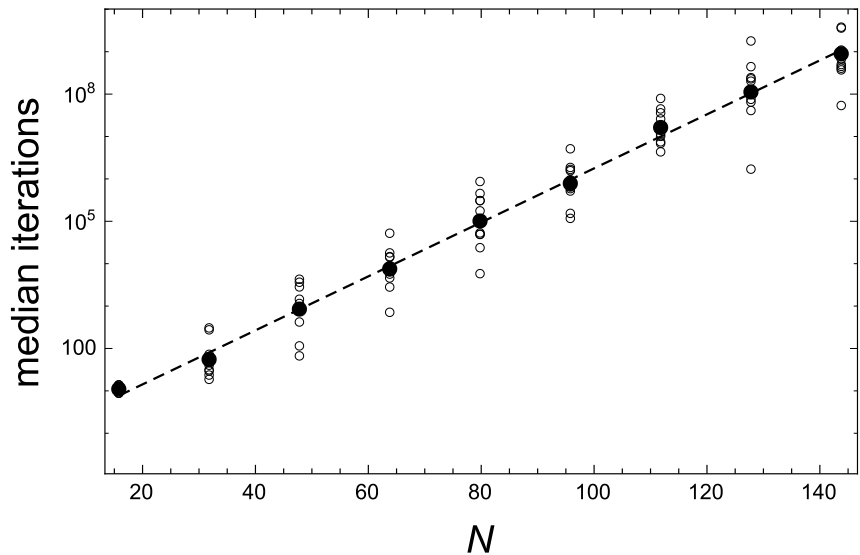}
\end{center}
\caption{Median iteration count of the RRR algorithm when solving average-case instances of bit retrieval. Results are shown for ten instances at each $N$ (open circles) as well as their geometric means (filled circles).}
\end{figure}

Figure 15 shows the much faster exponential growth for instances at the upper limit of the hardness index, $h\approx 1$. This study also compares the performance of the RRR algorithm with and without noise. In the instances with $N\equiv 3\pmod{4}$, where there exist perfect Hadamard sequences, we gave the noise-free autocorrelation data $a_k=-1, k\ne 0$ as input. In the other set of instances, $N\equiv 1\pmod{4}$, we specified the noisy autocorrelation $n_k=-1\pm 2, k\ne 0$ for solutions to be compatible with the symmetric Hadamard sequences. While in the second case it would have been easy to specialize the constraint projections for symmetric sequences, we chose not to in order to have noise be the only contrasting feature. The results suggest that noise has negligible effect on the complexity constant (slope): $c\approx 0.513$, $N\equiv 3$ and $c\approx 0.494$, $N\equiv 1$. The scatter of the data points about the straight lines in the plot is larger than our errors in estimating the median and therefore is intrinsic to each $N$. It is interesting that $N\equiv 1$ instances are easier by about a factor $180$ relative to the $N\equiv 3$ group.

\begin{figure}[!t]
\begin{center}
\includegraphics[width=4.5in]{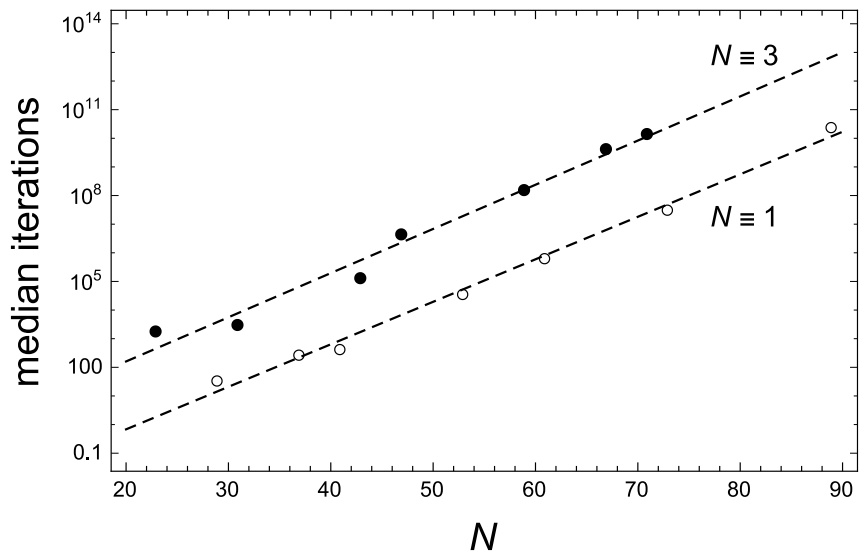}
\end{center}
\caption{Median iteration count of the RRR algorithm when solving Hadamard instances. Errors in the estimates of the median are smaller than the scatter of the points about the straight lines.}
\end{figure}

\section{Summary}\label{sec:summary}

We defined three versions of the bit retrieval problem. The noise-free problem arises in cryptographic attacks of digital signature schemes based on polynomial rings. While this is the easiest version, the best known algorithm still has subexponential complexity. With noise, even at a level that has no effect on solution uniqueness, the best algorithms have exponential complexity.
This is the version of bit retrieval that in x-ray crystallography is known as the phase problem.

Bit retrieval instances of the same size $N$ can have very different difficulty depending on the hardness index $h$. In the Howgrave-Graham-Szydlo algebraic algorithm the hardness index appears as the norm of the cyclotomic integer that must be factored (as an ordinary integer), while in constraint satisfaction algorithms it corresponds to the measure of a constraint set. Low-autocorrelation instances of bit retrieval, the most extreme form of which are solved by Hadamard sequences, have the highest hardness index.

The phase problem of x-ray crystallography normally corresponds to the sparse limit of bit retrieval. However, even sparse problems have a difficult regime, as quantified by the mean multiplicity $\mu$ of differences (atom-atom separations), in the cyclic difference set description of bit retrieval. Iterative constraint satisfaction algorithms are indispensable in this application of bit retrieval, having solved an estimated one million crystal structures \cite{S}. Even so, these algorithms have not received the theoretical scrutiny that normally goes with algorithms that provide a basic service to scientific investigations.

We have made a special effort to obtain the best estimates of the constant $c$ in the exponential complexity $2^{c N}$ achieved by algorithms that solve the noisy variant of bit retrieval. There is a simple algorithm that achieves $c=3/4$. For the branch and bound algorithm, based on convex relaxation of the non-convex constraints of bit retrieval, we obtain the numerical estimate $c\approx 0.564$ for the hardest ($h\approx 1$) instances. The best performance, $c\approx 0.504$, is achieved by the least understood algorithm, relaxed-reflect-reflect or RRR. While there is a vast literature on the application of algorithms like RRR to problems with convex constraints,
their continued success on highly non-convex problems is a curiosity that has not received the explanation it deserves.

In bit retrieval there is a simple argument why branch and bound, the RRR algorithm's closest competitor, is suboptimal. As in integer programming, branching is performed on a subset of the signs, with bounds arising from convex inequalities imposed on their complement. But unlike integer programming, in bit retrieval we are also forced to relax a second set of constraints, on the Fourier magnitudes, in order that the feasibility subproblems are convex and easily solved. The second relaxation does not change the feasible set of the complete problem (see section \ref{sec:convex}), but this manifests itself only once all the signs have been set to definite values. The branch and bound algorithm may thus follow branches that would have been discarded by an algorithm that did not need to attenuate the Fourier magnitude constraints.

The superior performance of RRR on noisy bit retrieval, relative to well established algorithms, is motivation for a better theoretical understanding of this algorithm. It is likely that the effectiveness of RRR, on hard non-convex problems, is only weakly linked to its properties in the convex domain. 
A possible alternative approach to its analysis was brought to light by our investigations of the flow limit ($\beta\to 0$). After correcting for the trivial `time-step' scaling, this limit optimizes performance. RRR search in the flow limit is confined to facets of the Voronoi cells of the discrete point set $B$, the hypercube of signs. A bias in the RRR sampling of $B$ that favors solutions is consistent with this facet attraction property.

\section{Acknowledgements}\label{sec:ack}

The author thanks the many individuals who have helped shape his understanding of bit retrieval over the years: J. Borwein, J. Buhler, C. Clement, P. Diaconis, N. Howgrave-Graham, G. Kuperberg, C. Moore, I. Rankenburg, J. Rosenberg, M. Szydlo. Support was provided by the Simons Foundation and DOE grant DE-FG02-11ER16210. Most of the work was carried out while the author was a visiting scientist at SLAC.

\section{Appendix}

\subsection{Hadamard sequences}\label{sec:hadamard}

For odd $N$ we define Hadamard sequences to be those $\pm 1$ sequences whose sum is 1 and whose $k\ne 0$ autocorrelations are as small as they can be. To expand on the last property, note that the identity
\begin{equation}
s_1 s_2\equiv s_1+s_2-1\pmod{4}
\end{equation}
for signs $s_1$ and $s_2$ implies that
\begin{equation}
a_k \equiv\sum_{l=0}^{N-1} s_l+s_{l-k}-1\equiv 2-N\equiv N\pmod{4}.
\end{equation}
Also, because $a_0=N$ and the sum of the autocorrelations is the square of the sum of the sign sequence, or 1, the average of the $k\ne 0$ autocorrelations must be $-1$. When $N\equiv 3\pmod {4}$ we can insist that $a_k=-1$, $k\ne 0$, while for $N\equiv 1\pmod {4}$ the best we can have is $a_k\in\{-3,1\}$, $k\ne 0$. The first class of Hadamard sequences is called perfect.

For either oddness of $N$ there is a simple construction of Hadamard sequences when $N$ is prime:
\begin{equation}
s_k=\left\{
\begin{array}{rl}
1,&k=0\\
\left(\frac{k}{N}\right),& k\ne 0.
\end{array}
\right.
\end{equation}
Because $s_k$ is explicitly given by the Legendre symbol, these are called Legendre sequences. Legendre sequences have the remarkable property \cite{H} that the $k\ne 0$ components of their Fourier transforms are simply obtained by applying a shift and, in one case, a complex rotation to the sequence itself:
\begin{equation}
\hat{s}_k=\left\{
\begin{array}{rl}
s_k+1/\sqrt{N},&N\equiv 1\pmod{4}\\
\mathrm{i}s_k+1/\sqrt{N},&N\equiv 3\pmod{4}.
\end{array}
\right.
\end{equation}
For $N\equiv 3$ the $k\ne 0$ Fourier magnitudes are perfectly equal while in the other case the magnitudes become uniform with increasing $N$. The $N\equiv 1$ Legendre sequences are symmetric because their Fourier transform is real.

In addition to the perfect Hadamard sequences given for prime $N$ by the Legendre symbol, there are also constructions whenever $N$ is the product of twin primes or one less than a power of two \cite{B}.

\subsection{Hardness distribution}\label{sec:harddist}

Freedman and Lane \cite{FL} proved that the distribution of the Fourier transform coefficients $\hat{s}_q$  of a sequence $s_0,\ldots,s_{N-1}$ of independently and identically distributed real random variables converges, for large $N$, to a distribution of independent and identical complex-normal distributions for $q=1,\ldots,M=\lfloor(N-1)/2\rfloor$. To characterize the distribution we therefore only need to compute the mean and variance of one of these Fourier coefficients. In the sequence distribution where $-1$ has density $\delta$,
\begin{eqnarray}
\langle\hat{s}_q\rangle_\delta&=&0,\\
\langle |\hat{s}_q|^2\rangle_\delta&=&\langle s_0^2\rangle_\delta+\sum_{k\ne 0}e^{i 2\pi k q/N}\langle s_0 s_{k}\rangle_\delta\\
&=&1-\langle s_0 s_{1}\rangle_\delta=4\delta(1-\delta).
\end{eqnarray} 
Moreover, since 
\begin{equation}
\log{h(s)}=\frac{1}{M}\sum_{q=1}^M\log{|\hat{s}_q|^2}
\end{equation}
is the average of $M$ independent and identically distributed real random variables, it has the central limit property in the limit of large $M$ when the mean and variance exist. To check the latter, we note that if $t=|\hat{s}_q|^2$ is the magnitude of a complex normal random variable of zero mean and variance $\sigma^2=4\delta(1-\delta)$, then $t$ has probability density
\begin{equation}
\rho(t)=\frac{e^{-t/\sigma^2}}{\sigma^2}.
\end{equation}
For large $N$ we therefore know that the log-hardness is concentrated at its mean value:
\begin{equation}
\langle \log{h(s)}\rangle_\delta=\int_0^\infty \rho(t) \log{t}\, dt= \log{\sigma^2}-\gamma.
\end{equation}

\subsection{Backtracking tree-width}\label{sec:maxwidth}

Applying Stirling's formula and the asymptotic conditions
\begin{equation}\label{asym}
1\ll k\ll N,\qquad 1\ll\mu
\end{equation}
to \eqref{n} we obtain
\begin{equation}
\log{n(k)}\sim -\frac{k^2}{2}e^{-\mu}+ k\log{N}-k\log{k}+k.
\end{equation}
We identify the maximum $k^*$ by the vanishing of the derivative with respect to $k$:
\begin{equation}
k^* e^{-\mu}\sim \log{N}-\log{k^*}\sim \log{N}.
\end{equation}
The maximum will be consistent with \eqref{asym},
\begin{equation}
k^*\sim e^{\mu}\log{N}\ll N
\end{equation}
provided we keep $\mu$ fixed (but large) as we take the limit $N\to \infty$.

\end{document}